\documentclass[journal, 10pt]{IEEEtran}

\usepackage{amsthm}
\usepackage{amsfonts}
\usepackage{amssymb}
\usepackage{mathrsfs}
\usepackage{mathtools}
\usepackage{float}
\usepackage[pdftex]{graphicx}
\usepackage{amsmath}
\usepackage{color}
\usepackage{multirow,multicol}
\usepackage{graphicx}
\usepackage{times}
\usepackage{textcomp}
\usepackage{verbatim}
\usepackage[table]{xcolor}
\usepackage{balance}
\usepackage{lipsum}
\usepackage[inline]{enumitem}
\usepackage{cuted}
\usepackage[caption=false,font=footnotesize]{subfig}
\usepackage{cite}
\usepackage{algpseudocode,algorithm}
\usepackage{hyperref}
\usepackage{tcolorbox}
\usepackage{setspace,epstopdf}
\usepackage[table]{xcolor}
\definecolor{color1}{RGB}{199,209,232}
\definecolor{color2}{RGB}{230,231,233}

\DeclareMathOperator*{\argmax}{argmax} 
\DeclareMathOperator*{\maximize}{maximize} 
\DeclareMathOperator*{\minimize}{minimize} 
\DeclareMathOperator*{\subjectto}{subject\hspace{3pt} to \hspace{3pt}} 
\newtheorem{theorem}{Theorem}

\newtheorem{lemma}[theorem]{Lemma}

\newtheorem{remark}{Remark}

\begin{document}
	
	\title{Spatial Path Index Modulation in mmWave/THz Band Integrated Sensing and Communications}
	
	
	\author{\IEEEauthorblockN{Ahmet M. Elbir, \textit{Senior Member, IEEE}, Kumar Vijay Mishra, \textit{Senior Member, IEEE},  \\
			Asmaa Abdallah, \textit{Member, IEEE}, Abdulkadir Celik, \textit{Senior Member, IEEE}, \\
			and Ahmed M. Eltawil, \textit{Senior Member, IEEE}  }
		
		\thanks{The conference precursor of this work has been presented in the 2023 IEEE Radar Conference~\cite{elbir_SPIM_MMWAVE_RadarConf_Elbir2022Nov}.
		}
		\thanks{A. M. Elbir is with  King Abdullah University of Science and Technology, Thuwal 23955, Saudi Arabia; and the Duzce University, Duzce 81620, Turkey (e-mail: ahmetmelbir@ieee.org).}
		\thanks{K. V. Mishra is with the United States DEVCOM Army Research Laboratory, Adelphi, USA (e-mail: kvm@ieee.org). }	
		\thanks{A. Abdallah, A. {C}elik and A. M. Eltawil are with King Abdullah University of Science and Technology, Thuwal 23955, Saudi Arabia (e-mail: asmaa.abdallah@kaust.edu.sa, abdulkadir.celik@kaust.edu.sa, ahmed.eltawil@kaust.edu.sa). } 
		
	}
	
	\maketitle


	\vspace{-64pt}
	\begin{abstract}
		As the demand for wireless connectivity continues to soar, the fifth generation and beyond wireless networks are exploring new ways to efficiently utilize the wireless spectrum and reduce hardware costs. One such approach is the integration of sensing and communications (ISAC) paradigms to jointly access the spectrum. Recent ISAC studies have focused on upper millimeter-wave and low terahertz bands to exploit ultrawide bandwidths. At these frequencies, hybrid beamformers that employ fewer radio-frequency chains are employed to offset expensive hardware but at the cost of lower multiplexing gains. Wideband hybrid beamforming also suffers from the beam-split effect arising from the subcarrier-independent (SI) analog beamformers. To overcome these limitations, we introduce a spatial path index modulation (SPIM) ISAC architecture, which transmits additional information bits via modulating the spatial paths between the base station and communications users. We design the SPIM-ISAC beamformers by estimating both radar and communications parameters through our proposed beam-split-aware algorithms. We then develop a family of hybrid beamforming techniques -- hybrid, SI, subcarrier-dependent analog-only, and beam-split-aware beamformers -- for SPIM-ISAC. Numerical experiments demonstrate that the proposed approach exhibits significantly improved spectral efficiency performance in the presence of beam-split when compared with even fully digital non-SPIM beamformers.

	\end{abstract}

	\vspace{-12pt}
	\section{Introduction}
	\label{sec:Introduciton}
	\IEEEPARstart{F}{or} several decades, radar and communications systems have exclusively operated in different frequency bands as allocated by the regulatory bodies to minimize the interference to each other~\cite{mishra2019toward}. Modern radar systems operate in various portions of the spectrum -- from very-high-frequency (VHF) to Terahertz ({THz})~\cite{elbir_thz_jrc_Magazine_Elbir2022Aug} -- for different applications, such as over-the-horizon, air surveillance, meteorological, military, and automotive radars \cite{mishra2023signal}. Similarly, communications systems have progressed from ultra-high-frequency (UHF) to millimeter-wave (mmWave) in response to the demand for new services, the massive number of users, and the applications with high data rate demands \cite{jrc_TCOM_Liu2020Feb,elbir2021JointRadarComm}. As a result, there has been substantial interest in designing \textit{integrated sensing and communications} (ISAC) systems that jointly access the scarce radio spectrum on an integrated hardware platform~\cite{mishra2019toward,liu2020co,duggal2020doppler,sedighi2021localization}.  In particular, as the allocation of the spectrum beyond 100 GHz is underway, ISAC is currently witnessing frantic research activity to simultaneously achieve high-resolution sensing and high data rate communications system architecture at both upper mmWave~\cite{mishra2019toward,thz_isac5_Petrov2019May} and low THz frequencies~\cite{elbir2021JointRadarComm,elbir_thz_jrc_Magazine_Elbir2022Aug}. 
	
	

	Signal processing at both mmWave and THz-band brings several new challenges, such as severe path loss, short transmission distance, and \textit{beam-split}~\cite{beamSquint_FeiFei_Wang2019Oct,delayPhasePrecoding_THz_Dai2022Mar,trueTimeDelayBeamSquint,elbir_THZ_CE_ArrayPerturbation_Elbir2022Aug}. To overcome these challenges at reduced hardware costs, hybrid analog and digital beamforming architectures are employed in a massive multiple-input multiple-output (MIMO) array configuration~\cite{heath2016overview,elbir2022Nov_Beamforming_SPM}. For higher spectral efficiency (SE) and lower complexity, massive MIMO systems employ wideband signal processing, wherein subcarrier-dependent (SD) baseband and subcarrier-independent (SI) analog beamformers are used. In particular, the weights of the analog beamformers are subject to a single (sub-)carrier frequency~\cite{alkhateeb2016frequencySelective}. Therefore, the beam generated across the subcarriers points to different directions causing {beam-split} (also referred to as beam-squint) phenomenon~\cite{elbir_THZ_CE_ArrayPerturbation_Elbir2022Aug,beamSquint_FeiFei_Wang2019Oct}. Compared to mmWave frequencies, the impact of beam-split is more severe in THz massive MIMO  because of wider system bandwidths in the latter (see Fig.~\ref{fig_ArrayGain_BS}). It is, therefore, highly desired to address beam-split for reliable system performance.

	The existing techniques to compensate for beam-split largely rely on additional hardware components, e.g., time-delayer (TD) \cite{trueTimeDelayBeamSquint,delayPhasePrecoding_THz_Dai2022Mar} and SD phase shifter networks~\cite{beamSquintAwareHB_SD_You2022Aug} to virtually realize SD analog beamformers. 
	However, these approaches are inefficient with respect to both cost and power~\cite{elbir_thz_jrc_Magazine_Elbir2022Aug}. Note that, the estimation of the communications channel and radar target direction-of-arrival (DoA) are handled in the digital domain. Hence, beam-split compensation for these tasks does not require additional hardware components. It is, therefore, possible to employ SD analog beamformers but the additional (analog) hardware is used only for hybrid (analog/digital) beamformer design. 
	
	Despite the cost-power benefits of hybrid analog/digital beamformers,
	they are limited in multiplexing gain~\cite{heath2016overview,elbir2022Nov_Beamforming_SPM}. This is of particular concern for future wireless communications, where improved energy/spectral efficiency (EE/SE) is a key consideration~\cite{heath2016overview}.  Lately, index modulation (IM) has attracted interest as a means to achieve improved EE and SE than the conventional modulation schemes~\cite{indexMod_Survey_Mao2018Jul,hodge2023index}. In IM, the transmitter encodes additional information in the indices of the transmission media such as subcarriers~\cite{hodge2020intelligent} (SIM), antennas~\cite{antenna_grouping_SM,hodge2019reconfigurable,jrc_spim_sm_Ma2021Feb}, and spatial paths~\cite{spim_bounds_JSTSP_Wang2019May,spim_BIM_TVT_Ding2018Mar,spim_GBM_Gao2019Jul,spim_onGSM_He2017Sep,spim_lowComplexGSM_Shi2021Jan,spim_GBMM_Guo2019Jul} (see Fig.~\ref{fig_IM_schemes}). {Besides, there are also techniques based on regulating the energy spectral densities of spatial waveforms~\cite{waveforDesign_SpectralMod1_Yu2022Apr,waveforDesign_SpectralMod2_Yang2022Dec}, {and IM-aided pilot design~\cite{Mao2021ISAC}. }}


	{
		\subsection{Motivation}
		In spatial modulation (SM), both antenna~\cite{spim_secureSM_SubarraySelection_Shu2020Nov} and path index~\cite{spim_secureGSM_Yang2020Oct,spim_secureGSM_FiniteAlphabet_Xia2020Dec} modulation can be performed.  Therefore, we categorize these SM techniques as spatial antenna/path index modulation (SAIM/SPIM), respectively. Specifically, in SAIM, the additional information bits are conveyed via switching the antennas. {However, a group of the antennas are silent at each transmission slot due to antenna switching,  leading to a reduction in beamforming gain compared to schemes that fully utilize all antennas.}  In contrast, SPIM enjoys employing all antennas while a switching network is used between the radio-frequency (RF) chains and the phase shifters, thereby, fully utilizing all antennas for beamforming gain. Furthermore, the implementation of SPIM does not require complex additional hardware since the SPIM process is governed by the channel path conditions.  Therefore, in this paper, we focus on SPIM. 
		
		Thanks to transmitting additional information bits via path indices, SPIM-aided systems have shown significant performance improvement in the communications-only systems. In particular, the SPIM-aided systems exhibit higher SE than that the use of fully digital beamformers in communications-only systems~\cite{spim_GBMM_Guo2019Jul}.  This creates the motivation to develop SPIM-based frameworks for THz systems in order to recover the SE loss due to beam-split.

		\subsection{Prior Work}
		SM techniques have been recently introduced to mmWave-MIMO communications systems while there are a few works on ISAC. Therefore, the prior works can be categorized as SM for  communications-only and ISAC systems, respectively.
		
		\subsubsection{SM for Communications-only Systems}
		SPIM-based communications scenario was considered, wherein the indices of the spatial paths were modulated to create different \emph{spatial patterns} for mmWave-MIMO in \cite{spim_BIM_TVT_Ding2018Mar}. The same approach was also exploited in~\cite{spim_GBM_Gao2019Jul} by employing lens arrays at both transmitter and receiver. In~\cite{spim_lowComplexGSM_Shi2021Jan}, a low-complexity approach was proposed for SPIM with a joint design of analog and digital beamformers. A similar architecture was also deployed for secure SAIM~\cite{spim_secureSM_SubarraySelection_Shu2020Nov} and SPIM~\cite{spim_secureGSM_Yang2020Oct,spim_secureGSM_FiniteAlphabet_Xia2020Dec} in the presence of eavesdropping users. Moreover, SE~\cite{spim_bounds_JSTSP_Wang2019May,spim_onGSM_He2017Sep} and EE~\cite{spim_EnergyE_Raafat2019Dec} have been utilized as performance metrics for analog-only beamforming and receiver design. In order to extract the spatial paths for SM, a super-resolution channel estimation approach was proposed in~\cite{spim_GSM_CE_Chu2019May}. In addition to the SM techniques employed over the antennas at the BS or communication user, SM over the reflecting surface elements has also been considered \cite{hodge2019reconfigurable,spim_IRS_SM_Yurduseven2020Aug,spim_IRS_SM_antennaSelection_Ma2020Jul}. Different from the aforementioned model-based techniques, a machine learning based approach was also proposed in~\cite{spim_FL_Elbir2021Jun} for SPIM.

		\begin{figure}
			\centering
			{\includegraphics[draft=false,width=1.\columnwidth]{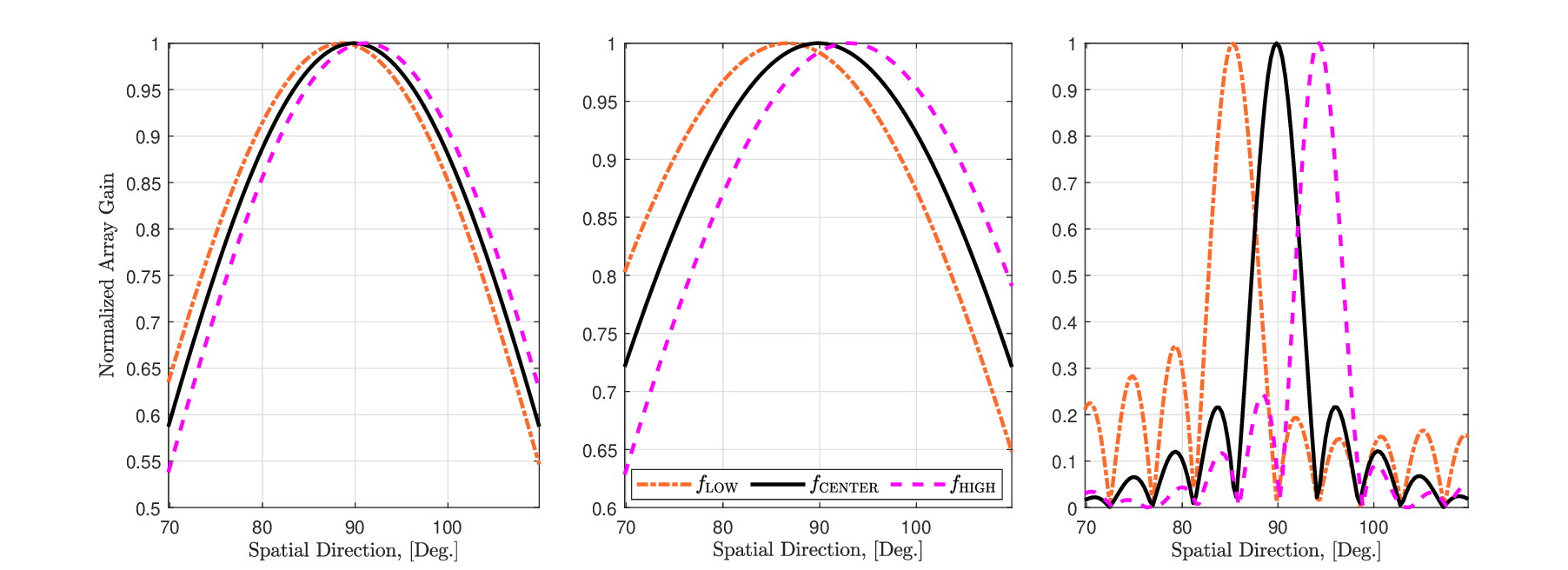} } 
			\caption{Normalized array gain with respect to spatial direction at the low, center, and high-end subcarriers for (left) $f_{\mathrm{CENTER}}=3.5$ GHz, $B=0.1$ GHz; (middle) $f_{\mathrm{CENTER}}=28$ GHz, $B=2$ GHz; and (right) $f_{\mathrm{CENTER}}=300$ GHz, $B=30$ GHz, respectively.    	}
			\label{fig_ArrayGain_BS}
		\end{figure}


				\begin{table*}[]
						\caption{Comparison With The State-of-The-Art 
							}
						\footnotesize
						\label{tableSummary}
						\centering
						\begin{tabular}{p{0.13\textwidth}p{0.08\textwidth}p{0.05\textwidth}p{0.08\textwidth}p{0.1\textwidth}p{0.05\textwidth}p{0.05\textwidth}p{0.05\textwidth}p{0.05\textwidth}}
								\hline 
								\hline
								\cellcolor{color2}\centering\arraybackslash q.v.  &\centering\arraybackslash mmWave \cellcolor{color1} & \cellcolor{color2}\centering\arraybackslash THz		& \cellcolor{color1}\centering\arraybackslash Wideband &\cellcolor{color2} \centering\arraybackslash Beam-split &\cellcolor{color1} \centering\arraybackslash SIM &\cellcolor{color2} \centering\arraybackslash SAIM &\cellcolor{color1} \centering\arraybackslash SPIM &  \cellcolor{color2} \centering\arraybackslash ISAC \\
								
								\hline
								\cellcolor{color2} \centering\arraybackslash \cite{spim_BIM_TVT_Ding2018Mar,spim_GBM_Gao2019Jul, spim_onGSM_He2017Sep,spim_bounds_JSTSP_Wang2019May,spim_lowComplexGSM_Shi2021Jan,spim_secureGSM_Yang2020Oct,spim_secureGSM_FiniteAlphabet_Xia2020Dec}
								& \cellcolor{color1} \centering \checkmark 
								& \cellcolor{color2} \centering $\times$ 
								& \cellcolor{color1} \centering $\times$
								& \cellcolor{color2} \centering $\times$
								& \cellcolor{color1} \centering $\times$
								& \cellcolor{color2} \centering $\times$
								& \cellcolor{color1}\centering \checkmark
								& \cellcolor{color2} \centering\arraybackslash  $\times$  \\
								\hline
								\cellcolor{color2} \centering\arraybackslash \cite{hodge2020intelligent,hodge2023index}
								& \cellcolor{color1} \centering \checkmark 
								& \cellcolor{color2} \centering $\times$ 
								& \cellcolor{color1} \centering \checkmark
								& \cellcolor{color2} \centering $\times$
								& \cellcolor{color1} \centering \checkmark
								& \cellcolor{color2} \centering $\times$
								& \cellcolor{color1}\centering  $\times$
								& \cellcolor{color2} \centering\arraybackslash  $\times$  \\
								\hline
								\cellcolor{color2} \centering\arraybackslash \cite{jrc_spim_sm_Ma2021Feb,spim_SM_Clutter_Zhang2021May,spim_SM_Clutter2_Zhang2021May,spim_SAIM_LowComple_Zhu2021Jun}
								& \cellcolor{color1} \centering \checkmark 
								& \cellcolor{color2} \centering $\times$ 
								& \cellcolor{color1} \centering $\times$
								& \cellcolor{color2} \centering $\times$
								& \cellcolor{color1}\centering $\times$
								& \cellcolor{color2} \centering $\times$
								& \cellcolor{color1} \centering $\times$
								& \cellcolor{color2} \centering\arraybackslash  \checkmark  \\
								\hline
								\cellcolor{color2} \centering\arraybackslash \cite{spim_secureSM_SubarraySelection_Shu2020Nov}
								& \cellcolor{color1} \centering \checkmark 
								& \cellcolor{color2} \centering $\times$ 
								& \cellcolor{color1} \centering $\times$
								& \cellcolor{color2} \centering $\times$
								& \cellcolor{color1}\centering $\times$
								& \cellcolor{color2} \centering \checkmark
								& \cellcolor{color1} \centering $\times$
								& \cellcolor{color2} \centering\arraybackslash  $\times$  \\
								\hline
								\cellcolor{color2} \centering\arraybackslash \cite{jrc_generalized_SM_Xu2020Sep}
								& \cellcolor{color1} \centering \checkmark 
								& \cellcolor{color2} \centering $\times$ 
								& \cellcolor{color1} \centering \checkmark
								& \cellcolor{color2} \centering $\times$
								& \cellcolor{color1}\centering $\times$
								& \cellcolor{color2} \centering \checkmark
								& \cellcolor{color1} \centering $\times$
								& \cellcolor{color2} \centering\arraybackslash  \checkmark \\
								\hline
								\cellcolor{color2} \centering\arraybackslash \cite{spim_SAIM_LowComple_Zhu2021Jun}
								& \cellcolor{color1} \centering \checkmark 
								& \cellcolor{color2} \centering $\times$ 
								& \cellcolor{color1} \centering $\times$
								& \cellcolor{color2} \centering $\times$
								& \cellcolor{color1} \centering $\times$
								& \cellcolor{color2} \centering $\times$
								& \cellcolor{color1}\centering \checkmark
								& \cellcolor{color2} \centering\arraybackslash  $\times$  \\
								\hline
								\cellcolor{color2} \centering\arraybackslash \cite{elbir2021JointRadarComm}
								& \cellcolor{color1} \centering $\times$ 
								& \cellcolor{color2} \centering \checkmark 
								& \cellcolor{color1} \centering \checkmark
								& \cellcolor{color2} \centering \checkmark
								& \cellcolor{color1}\centering $\times$ 
								& \cellcolor{color2} \centering $\times$
								& \cellcolor{color1} \centering $\times$
								& \cellcolor{color2} \centering\arraybackslash  \checkmark  \\
								\hline
								\cellcolor{color2} \centering\arraybackslash \cite{elbir_SPIM_MMWAVE_RadarConf_Elbir2022Nov}
								& \cellcolor{color1} \centering \checkmark 
								& \cellcolor{color2} \centering $\times$ 
								& \cellcolor{color1} \centering $\times$
								& \cellcolor{color2} \centering $\times$
								& \cellcolor{color1} \centering $\times$
								& \cellcolor{color2} \centering $\times$
								& \cellcolor{color1}\centering \checkmark
								& \cellcolor{color2} \centering\arraybackslash  \checkmark  \\
								\hline
								\cellcolor{color2} \centering\arraybackslash This paper
								& \cellcolor{color1} \centering \checkmark 
								& \cellcolor{color2} \centering \checkmark
								& \cellcolor{color1} \centering \checkmark
								& \cellcolor{color2} \centering \checkmark
								& \cellcolor{color1} \centering $\times$
								& \cellcolor{color2} \centering $\times$
								& \cellcolor{color1}\centering \checkmark
								& \cellcolor{color2} \centering\arraybackslash  \checkmark \\
								\hline
								\hline
							\end{tabular}
					\end{table*}

		\begin{figure*}
			\centering
			{\includegraphics[draft=false,width=.99\textwidth]{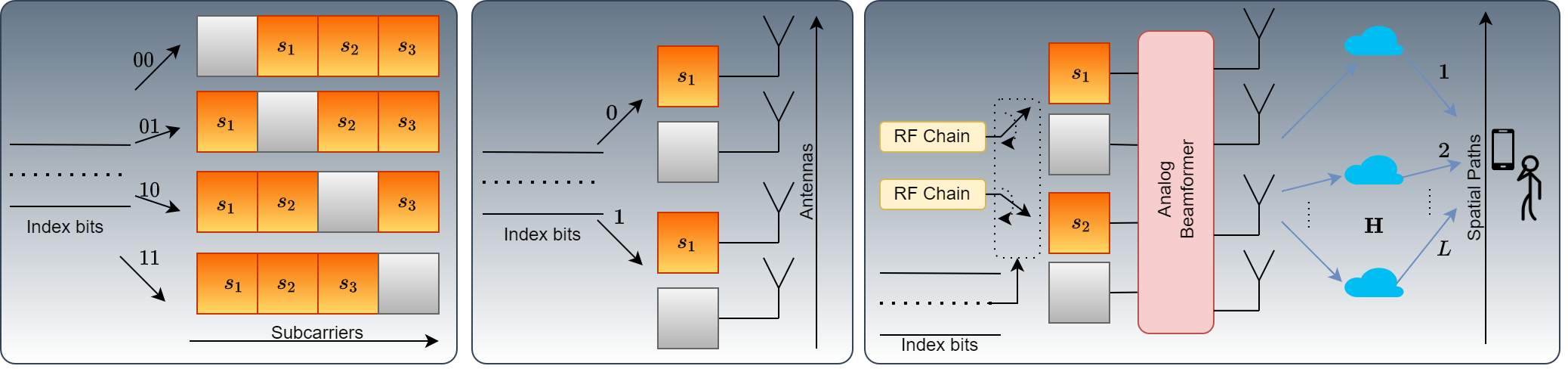} } 
			\caption{IM over subcarriers (left), antennas (middle), and spatial path indices (right).  }
			\label{fig_IM_schemes}
		\end{figure*}

		\subsubsection{SM for ISAC systems}
		Although there are several SM-based studies in the literature for communication-only systems, its usage for ISAC applications is relatively recent.  For SM-aided ISAC systems, \cite{jrc_spim_sm_Ma2021Feb} devised a SAIM approach, wherein the antenna subarrays are allocated between different radar pulses and symbol time slots to handle sensing and communications tasks disjointly without mutual interference.  This was further investigated in~\cite{jrc_generalized_SM_Xu2020Sep}, which employed SM over antenna indices for orthogonal frequency division multiplexing (OFDM) ISAC, for which the OFDM carriers are divided into two groups and assigned exclusively to an active antenna to perform sensing and communications. In a more sensing-centric scenario, \cite{spim_SM_Clutter_Zhang2021May,spim_SM_Clutter2_Zhang2021May} proposed a clutter suppression approach for ISAC based on the similarity of the generated spatial patterns. Also in~\cite{spim_SAIM_LowComple_Zhu2021Jun}, a low-complexity SAIM technique was proposed with one-bit analog-digital converters (ADCs). 
		
		To sum up, the aforementioned ISAC works~\cite{jrc_spim_sm_Ma2021Feb,spim_SM_Clutter_Zhang2021May,spim_SM_Clutter2_Zhang2021May,spim_SAIM_LowComple_Zhu2021Jun,jrc_generalized_SM_Xu2020Sep} consider only SM over antenna or subcarrier indices and do not exploit SPIM (see Table~\ref{tableSummary}).
		On the other hand, the proposed SPIM approaches in~\cite{spim_BIM_TVT_Ding2018Mar,spim_GBM_Gao2019Jul, spim_onGSM_He2017Sep,spim_bounds_JSTSP_Wang2019May} consider the communications-only scenario without accounting for the trade-off between sensing and communications functionalities. 
		
		
	}
	

	
	

	
	\subsection{Our Contributions}
	{Contrary to the aforementioned studies, we leverage SPIM for mmWave and THz ISAC systems to achieve higher spectral efficiency. Preliminary results of our work appeared in our conference publication \cite{elbir_SPIM_MMWAVE_RadarConf_Elbir2022Nov}, where only a single target scenario was investigated for narrowband mmWave system. In this paper, we expand our study to include wideband THz-band systems that are susceptible to beam-split. We also propose novel hybrid beamforming techniques to mitigate the impact of beam-split.  Note that the algorithms proposed here are also applicable to both narrow and wideband mmWave systems.} Our main contributions are:
	\begin{enumerate}[wide]
		\item \textbf{SPIM-ISAC:} Despite the performance loss resulting from beam-split in wideband systems (especially at THz-band), our proposed SPIM-ISAC approach is particularly helpful in improving the SE through the transmission of additional information bits. By exploiting the SPIM, our proposed approach surpasses even fully digital (FD) beamformer design, thereby exhibiting great potential for the next generation of sensing and communications systems. {To this end, we first perform an estimation of ISAC parameters that includes the directions of targets and the user paths. Then, we construct the SPIM analog beamformer by utilizing the steering vectors corresponding to these parameters.   }
		
		{
			\item \textbf{Beam-Split-Aware (BSA) Algorithms:} We introduce efficient approaches for ISAC parameter estimation and hybrid beamforming while simultaneously compensating the impact of beam-split without additional hardware components such as TDs. In particular, for radar target parameters, we propose a BSA \textit{mu}ltiple \textit{si}gnal \textit{c}lassification (MUSIC) algorithm, wherein the DoAs are accurately estimated from the beam-split-corrupted wideband array data. For communications parameters,  we adapt BSA orthogonal matching pursuit (OMP) from \cite{elbir_BSA_OMP_THZ_CE_Elbir2023Feb} for the SPIM scenario. Our BSA-MUSIC and BSA-OMP algorithms take into account the angular deviation from beam-split, thereby, \textit{ipso facto} mitigating the effect of beam-split. Unlike prior works that rely on TD networks, our approach does not require additional hardware while still yielding satisfactory performance. {While beam-split has recently been investigated in mmWave and THz ISAC scenarios~\cite{elbir2021JointRadarComm,beamSquintAwareHB_SD_You2022Aug,Gao2023BSISAC,elbir_ISAC_antennaSelection_THZ_Elbir2023Jul}, IM-based transceiver architecture is not considered. }
			
		}
		
		\item \textbf{Analog-Only and Hybrid Beamformers:} 
		We also propose three different beamformers: SD-analog-only (AO), SI-AO, and hybrid (analog/digital). 
		While the SD-AO beamformer accurately compensates for the beam-split, its hardware employs SD phase shifter networks and is more complex. The SI-AO beamformer yields a simpler architecture but at the cost of lower SE. {Finally, for hybrid beamformer, we design an updated \textit{BSA baseband beamformer} that compensates for the beam-split in the baseband. Our proposed SPIM-ISAC hybrid beamformer achieves higher SE than the AO beamformers yet exhibit lower hardware complexity.} The SPIM-ISAC beamformer simultaneously maximizes the SE at the communications user over SPIM-aided signaling and achieves as much signal-to-noise ratio (SNR) as possible for detecting the radar targets. The SPIM-ISAC analog beamformer comprises radar-only and communications-only beamformers. While the former is constructed from the steering vectors corresponding to the target DoAs, the latter is selected from different spatial patterns between the BS and the communications user. The proposed design also includes a trade-off parameter between communications and radar sensing operations in the sense that the SNR at the targets and the user is controlled. 
	\end{enumerate}

	
	{
		\subsection{Paper Outline}
		The rest of the paper is organized as follows. In the next section, we describe the signal and system model for SPIM-ISAC. In Section~\ref{sec:ProbFormulation}, we formulate the hybrid beamforming problem for SPIM-ISAC. In Section~\ref{sec:ParameterEst}, we develop techniques to estimate radar and communications channel parameters. Next, we introduce our SPIM-ISAC framework as well as the proposed beamforming algorithms in Section~\ref{sec:SPIMISAC}. We validate our model and methods through numerical experiments in Section~\ref{sec:Sim} and conclude in Section~\ref{sec:Conc}. 
	}
	\subsection{Notation}
	Throughout this paper,  $(\cdot)^\textsf{T}$ and $(\cdot)^{\textsf{H}}$ denote the transpose and conjugate transpose operations, respectively. For a matrix $\mathbf{A}$ and vector $\mathbf{a}$; $[\mathbf{A}]_{ij}$, $[\mathbf{A}]_k$  and $[\mathbf{a}]_l$ correspond to the $(i,j)$-th entry, $k$-th column and $l$-th entry, respectively. $\lfloor\cdot \rfloor$ and $\mathbb{E}\{\cdot\}$ represent the flooring and expectation operations, respectively. The binomial coefficient is defined as $\footnotesize \left(\begin{array}{c}
		n\\
		k
	\end{array}  \right) = \frac{ n!}{k! (n-k)!}$. An $N\times N$ identity matrix is represented by $\mathbf{I}_{N} $. The pulse-shaping function is represented by {$\mathrm{sinc}(t) = \frac{\sin t}{t}$,} $\xi (a) = \frac{\sin N \pi a}{N \sin \pi a }$ is the Dirichlet sinc function, and $\lceil \cdot \rceil$ denotes the ceiling operation. We denote $|| \cdot||_2$ and $|| \cdot||_\mathcal{F}$ as the  $l_2$-norm and Frobenious norm, respectively.

	\section{System Model}
	\label{sec:Model}
	\begin{figure}[t]
		\centering		{\includegraphics[draft=false,width=.99\columnwidth]{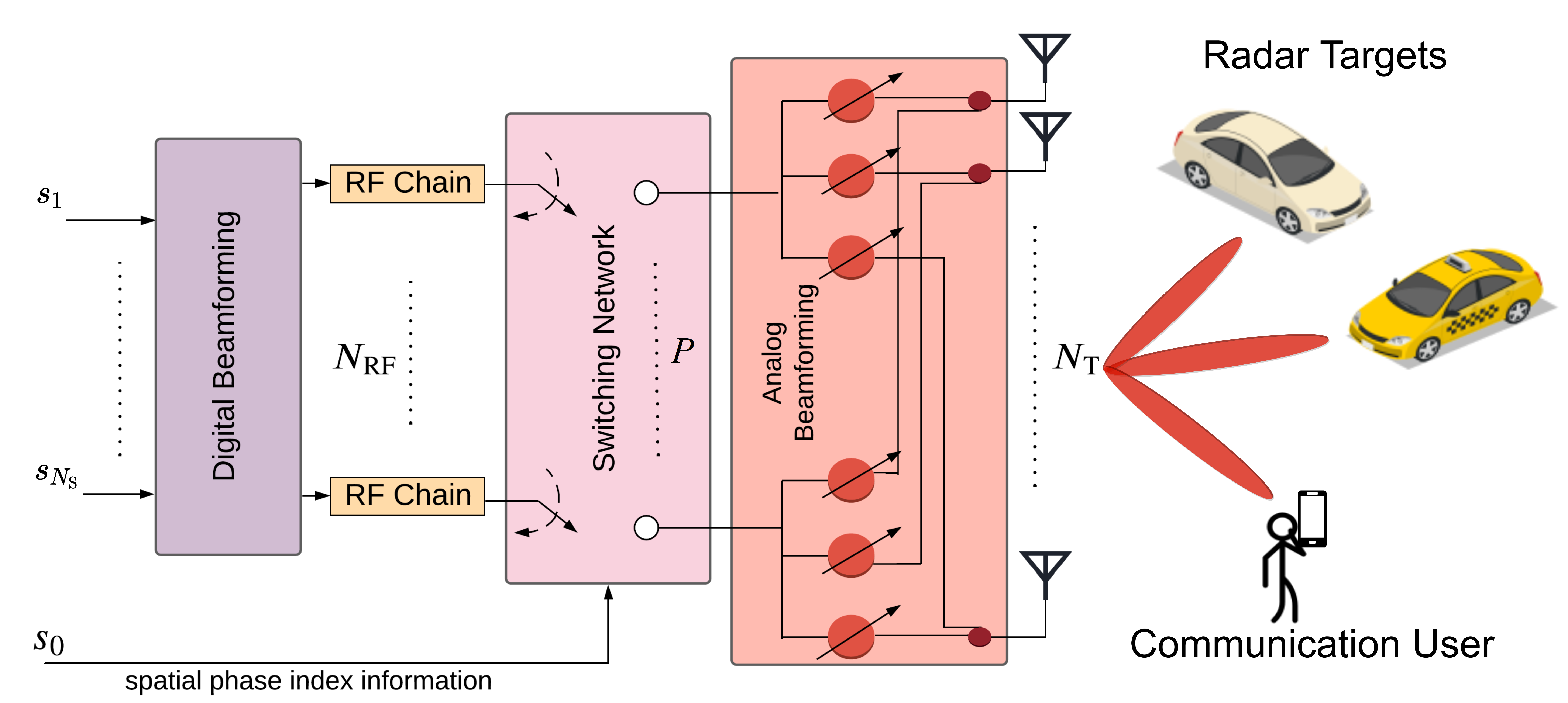} }
		\caption{The SPIM-ISAC architecture processes the incoming data streams and employs spatial path index information $s_0$ in a switching network, which connects $N_\mathrm{RF}$ RF chains to $P= L + K$ taps on the analog beamformers to exploit $K$ paths for the radar targets and $L_\mathrm{S}$ out of $L$ spatial paths for the communications user. 
		}
		\label{fig_BS}
	\end{figure}
	Consider a wideband transmitter design problem in an ISAC scenario with SPIM. The scenario involves {a single user} and $K$ radar targets (Fig.~\ref{fig_BS}). The dual-function BS  employs $M$ subcarriers and it  has $N_\mathrm{T}$ antennas to jointly communicate with the communication user and sense the radar targets via probing signals. The user has $N_\mathrm{R}$ antennas, for which $N_\mathrm{S}$ data symbols $\mathbf{s}[m] = [s_1[m],\cdots,s_{N_\mathrm{S}}[m]]^\textsf{T}\in \mathbb{C}^{N_\mathrm{S}}$ are transmitted, where $\mathbb{E}\{\mathbf{s}[m]\mathbf{s}^\textsf{H}[m]\}=\frac{1}{N_\mathrm{S}}\mathbf{I}_{N_\mathrm{S}}$.  Additionally, the spatial path index information represented by ${s}_0$ is fed to the switching network (Fig.~\ref{fig_BS}) to  assign the outputs of $N_\mathrm{RF}$ RF chains to the $ P$ taps of the analog beamformer. {Here, $P$ is a priori parameter and it is defined as the total number of \textit{spatial paths} resolved at the BS from the targets and communications user as $	P = L + K,$}
	where $L$ and $K$ paths are reserved for communications and radar operations, respectively.	We also define $L_\mathrm{S}$ as the selected number of paths out of $L$ total communication paths for SPIM.	Thus, we have $	N_\mathrm{RF} = L_\mathrm{S} + K,$
	which indicates that $L_\mathrm{S}$ columns of the analog beamformer are dedicated to communications task while the remaining $K$ columns are employed for sensing.
	As a result of this \textit{information-driven random switching} with IM~\cite{indexMod_Survey_Mao2018Jul}, there exist $\footnotesize \left(\begin{array}{c}
		L \\
		Ls
	\end{array}  \right)$ choices of connection to incorporate the spatial domain information as a principle of IM.  Thus, the BS can process at most $P\geq N_\mathrm{RF}$ inputs, and each of the $P$ inputs is connected to the $N_\mathrm{T}> P$ BS antennas via phase shifters forming a fully-connected structure. The switching operation between the RF chains and the phase shifters needs to be performed in accordance with the symbol duration, for which low-cost switches with the speed of nanoseconds are available~\cite{spim_GBMM_Guo2019Jul,heath2016overview}.
	
	Thus, compared to the conventional ISAC systems, the proposed SPIM-ISAC architecture has the advantage of transmitting additional data streams toward the communication user by exploiting the \textit{spatial pattern} of the channel with limited RF chains, i.e., $N_\mathrm{RF}\leq P$ while performing radar sensing task with $K$ line-of-sight (LoS) spatial paths. Note also that the communication user requires  $\bar{N}_\mathrm{RF} \geq L_\mathrm{S}$ RF chains in order to perform SPIM for processing $L_\mathrm{S}$ paths.  Then, we define the  total number of spatial patterns for communication as $	S = 2^{\left\lfloor \log_2 \left(\footnotesize\begin{array}{c}
			L \\
			L_\mathrm{S}
		\end{array}  \right) \right \rfloor  }.$

	\begin{remark}     
		In the limiting cases of $K=0$ (communications-only) and $L=0$ (sensing-only), the proposed ISAC approach remains viable regardless of the received paths from the user and targets. This is managed by adjusting the sensing-communications trade-off parameter. Further, when $L_\mathrm{S} = L$, the proposed SPIM-ISAC configuration reduces to a conventional ISAC system, regardless of  $K$, because there is only one choice of connection of transmission~\cite{spim_bounds_JSTSP_Wang2019May,spim_GBMM_Guo2019Jul}.\end{remark} 
	
	{
		\begin{remark}
			Since  $P$ is an environment-dependent parameter, the proposed switching network architecture at the BS as shown in Fig.~\ref{fig_BS} implies that SPIM can be performed for at most $P$ spatial paths.  Suppose there are $\bar{P}$ paths in the environment, then, $P$ is selected as
			\begin{align}
				P = \left\{ \begin{array}{cc}
					{P} &   \text{if } P \leq \bar{P}  \\
					\bar{P}&  \text{otherwise}.
				\end{array} \right.   
			\end{align}
		\end{remark} 
	}			

	\subsection{Communications Model}
	Consider the BS that aims to transmit the data symbol vector $\mathbf{s}[m]\in \mathbb{C}^{N_\mathrm{S}}$ toward the communications user. The BS first applies the SD baseband beamformer $\mathbf{F}_\mathrm{BB}^{(i)}[m]\in\mathbb{C}^{N_\mathrm{RF}\times N_\mathrm{S}}$ ($N_\mathrm{S} = L_\mathrm{S}$) for the $i$-th spatial pattern. Then,  $M$-point inverse fast Fourier transform (IFFT) is applied to convert the signal to time-domain and then appended with the cyclic prefix (CP). Finally, the SI analog beamformer ${\mathbf{F}}_\mathrm{RF}^{(i)} \in \mathbb{C}^{N_\mathrm{T}\times N_\mathrm{RF}}$ is applied. Denoting the index set of possible spatial patterns by $\mathcal{S} = \{1,\cdots, S\}$, the $N_\mathrm{T}\times 1$ transmit signal for the $i$-th,  ($i\in \mathcal{S}$), the spatial pattern becomes $	\mathbf{x}^{(i)}[m] = \mathbf{F}_\mathrm{RF}^{(i)}\mathbf{F}_\mathrm{BB}^{(i)}[m]\mathbf{s}[m],$
	where the analog beamformer $\mathbf{F}_\mathrm{RF}^{(i)}$ has constant-modulus constraint, i.e., $|[\mathbf{F}_\mathrm{RF}^{(i)}]_{n,r}| = 1/\sqrt{N_\mathrm{T}}$ for $n = 1,\cdots, N_\mathrm{T}$, $r = 1,\cdots,N_\mathrm{RF}$. Further, we have $\sum_{m=1}^{M}\| \mathbf{F}_\mathrm{RF}^{(i)}\mathbf{F}_\mathrm{BB}^{(i)}[m]\|_\mathcal{F}^2 = MN_\mathrm{S}$ to account for the total power constraint.


	\subsubsection{Channel} We employ Saleh-Valenzuela (S-V) multipath channel model, which is the superposition of received non-LoS (NLoS) paths to model both mmWave and THz channels~\cite{ummimoTareqOverview,ummimoHBThzSVModel,alkhateeb2016frequencySelective,heath2016overview}.	Compared to the mmWave channel, the THz channel involves limited reflected paths and negligible scattering~\cite{ummimoTareqOverview,thz_mmWave_path_Comparison_Yan2020Jun}. For example, approximately $5$ paths survive at $0.3$ THz for THz massive MIMO systems as compared to approximately $8$ paths at $60$ GHz~\cite{thz_mmWave_path_Comparison_Yan2020Jun}. Especially for outdoor applications, multipath channel models are widely used to represent the THz channel for a more general scenario~\cite{ummimoTareqOverview,ummimoHBThzSVModel,thz_mmWave_path_Comparison_Yan2020Jun}.  
	Consider the delay-$\bar{d}$ $N_\mathrm{R}\times N_\mathrm{T}$ MIMO communications channel involving $L$ NLoS paths in discrete-time domain as
	\begin{align}
		\label{channelTimeDomain}
		\tilde{\mathbf{H}}(\bar{d}) = \sum_{l = 1}^{L} {\gamma}_l \mathrm{sinc}(\bar{d} - B\tau_l) \mathbf{a}_\mathrm{R}(\phi_l) \mathbf{a}_\mathrm{T}^\textsf{H}(\theta_l), 
	\end{align}
	where ${\gamma}_l\in \mathbb{C}$ denotes the channel path gain, $B$ represents the system bandwidth and $\tau_l$ is the time delay of the $l$-th path.  $\phi_l$ and $\theta_l$ denote the physical DoA and direction-of-departure (DoD) angles of the scattering paths between the user and the BS, respectively, where  $\phi_l = \sin \tilde{\phi}_l$, $\theta_l = \sin \tilde{\theta}_l$ and $\tilde{\phi}_l,\tilde{\theta}_l \in [-\frac{\pi}{2},\frac{\pi}{2}]$. Then, the corresponding receive and transmit steering vectors are defined as $\mathbf{a}_\mathrm{R}(\phi_l)\in \mathbb{C}^{N_\mathrm{R}}$ and $\mathbf{a}_\mathrm{T}(\theta_l)\in \mathbb{C}^{N_\mathrm{T}}$, respectively. Performing $M$-point FFT of the delay-$\bar{d}$ channel given in (\ref{channelTimeDomain}) yields $\mathbf{H}[m] = \sum_{\bar{d}=1}^{\bar{D}-1} \tilde{\mathbf{H}}(\bar{d}) e^{- \mathrm{j}\frac{2\pi m}{M} \bar{d} }  $, where $\bar{D}\leq M$ is the CP length. Then, the $N_\mathrm{R}\times N_\mathrm{T}$ channel matrix in frequency domain is 
	\begin{align}
		\label{channelFrequencyDomain}
		\mathbf{H}[m] = \sum_{l = 1}^{L} {\gamma}_l \mathbf{a}_\mathrm{R}(\varphi_{l,m}) \mathbf{a}_\mathrm{T}^\textsf{H}(\vartheta_{l,m})e^{-\mathrm{j}2\pi \tau_l f_m},
	\end{align}
	where $\varphi_{l,m}$ and $\vartheta_{l,m}$ denote the spatial directions, which are SD and they are deviated from the physical directions $\phi_l$, $\theta_l$ in the beamspace because of the beam-split. On the other hand, the beam-split-free channel matrix is $\overline{\mathbf{H}}[m] = \sum_{l = 1}^{L} {\gamma}_l \mathbf{a}_\mathrm{R}(\phi_{l}) \mathbf{a}_\mathrm{T}^\textsf{H}(\theta_{l})e^{-\mathrm{j}2\pi \tau_l f_m}$.
	\subsubsection{Beam-split}
	\label{sec:beamSplit} 
	
	In wideband transmission, a single wavelength assumption, i.e., $\lambda_1 = \cdots =\lambda_M = \frac{c_0}{f_c}$, is usually made across the subcarriers, where $c_0$ and $f_c$ are the speed of light and carrier frequency, respectively. However, when a common analog beamformer is used, the single wavelength assumption does not hold; the generated beams 
	split and spatially point to different directions \cite{beamSquint_FeiFei_Wang2019Oct,elbir_thz_jrc_Magazine_Elbir2022Aug}. Suppose that similar beamforming architecture (i.e., SI analog beamformer with SD digital beamformers) is employed by the user. Then, the DoA angles at the user are also affected by beam-split.   The relationship between the spatial  ($\varphi_{l,m},\vartheta_{l,m}$) and the physical directions ($\phi_l, \theta_l$)  is given as
	\begin{align}
		\varphi_{l,m} =\eta_m \phi_l, \hspace{10pt}
		\vartheta_{l,m} = \eta_m\theta_l, \label{beamSplitforAngles}
	\end{align}	
	where $\eta_m = \frac{f_m}{f_c}$, $f_m = f_c + \frac{B}{M}(m - 1 - \frac{M-1}{2})$ is the $m$-th subcarrier frequency for the system bandwidth $B$.    
	
	The beam-split is mitigated if the spatial $\varphi_{l,m},\vartheta_{l,m}$ and physical directions $\phi_l, \theta_l$ are equal, i.e., $\eta_m = 1$. This is accomplished via any of the following methods:
	\begin{enumerate}
		\item Narrowband scenario, wherein $f_m \approx f_c$ and the carrier frequency is much larger than the system bandwidth (i.e. $f_c \gg B$).
		\item Using additional hardware components, e.g., TDs, (each of which consumes approximately $100$ mW~\cite{elbir_thz_jrc_Magazine_Elbir2022Aug,delayPhasePrecoding_THz_Dai2022Mar}) between the phase shifters and the RF chains~\cite{trueTimeDelayBeamSquint,delayPhasePrecoding_THz_Dai2022Mar} to compensate for the angular deviation in the generated beams due to beam-split via generating virtual SD analog beamformers.
		\item Designing SD analog beamformers (see Sec.~\ref{sec:AOBF}), which can alleviate the effect of beam-split at the cost of employing $MN_\mathrm{T}N_\mathrm{RF}$ (instead of $N_\mathrm{T}N_\mathrm{RF}$) phase-shifters (each of which consumes approximately $20$ mW at $60$ GHz ($40$ mW at $0.3$ THz)~\cite{elbir_thz_jrc_Magazine_Elbir2022Aug}). 
		\item Advanced signal processing techniques to compensate for the beam-split by correcting the deviated phase terms of the analog beamformers (see Sec.~\ref{sec:BeamSplitMitigation}).
	\end{enumerate}

	Consider  a uniform linear array (ULA) configuration with $d = \frac{\lambda_c}{2}  = \frac{c_0}{2f_c}$  half-wavelength element spacing. The $n$-th element of the beam-split-free transmit steering vector is $[\mathbf{a}_\mathrm{T}(\theta_l)]_n = \frac{1}{N_\mathrm{T}} \exp \{- \mathrm{j} \pi (n-1)\theta_l  \}$. However, under the effect of beam-split, the $n$-th entry of the SD steering vector $\mathbf{a}_\mathrm{T}(\vartheta_{l,m})$ is 
	\begin{align}
		[\mathbf{a}_\mathrm{T}(\vartheta_{l,m})]_n &= \frac{1}{\sqrt{N_\mathrm{T}}} \exp \left\{- \mathrm{j} \frac{2\pi d}{\lambda_m } (n-1) \theta_l\right \} 
		\nonumber\\	& 
		=\frac{1}{\sqrt{N_\mathrm{T}}} \exp\left\{- \mathrm{j}\pi  \frac{ f_m}{f_c }(n-1) \theta_l \right\} \nonumber \\
		&=\frac{1}{\sqrt{N_\mathrm{T}}} \exp\left\{- \mathrm{j}\pi (n-1)\eta_m \theta_l \right\},\label{steringVec_aT}
	\end{align}
	where    $\lambda_m = \frac{c_0}{f_m}$ is the wavelength of the $m$-th subcarrier. Note that $\eta_m = 1$ in (\ref{steringVec_aT}) implies zero beam-split. The  channel model in (\ref{channelFrequencyDomain}) in a compact form is $	\mathbf{H}[m] = \mathbf{P}_m\boldsymbol{\Lambda}_m\mathbf{Q}_m^\textsf{H},$
	where the matrices $\mathbf{P}_m \in \mathbb{C}^{N_\mathrm{R}\times L}$ and $\mathbf{Q}_m\in \mathbb{C}^{N_\mathrm{T}\times L}$ represent the receive and transmit array responses for $L$ paths as $\mathbf{P}_m = [\mathbf{a}_\mathrm{R}(\varphi_{1,m}),\cdots, \mathbf{a}_\mathrm{R}(\varphi_{L,m})]$ and $\mathbf{Q}_m = [\mathbf{a}_\mathrm{T}(\vartheta_{1,m}),\cdots,\mathbf{a}_\mathrm{T}(\vartheta_{L,m})]$, respectively; $\boldsymbol{\Lambda}_m\in \mathbb{C}^{L \times L}$ is a diagonal matrix comprised of path gains ${\gamma}_l$ as $	\boldsymbol{\Lambda}_m = \mathrm{diag}\{\tilde{\gamma}_{1},\cdots, \tilde{\gamma}_L\}$,	where $\tilde{\gamma}_l = {\gamma}_le^{- \mathrm{j}2\pi \tau_l f_m}$ and ${\gamma}_{1} > {\gamma}_{2} > \cdots > {\gamma}_L$. Then, the $N_\mathrm{R}\times 1$ received signal at the communications user for the  $i$-th spatial pattern is $	\mathbf{H}[m] = \mathbf{P}_m\boldsymbol{\Lambda}_m\mathbf{Q}_m^\textsf{H},$
	where $\mathbf{n}[m]\sim \mathcal{CN}(\mathbf{0},\sigma_n^2\mathbf{I}_{N_\mathrm{R}})\in\mathbb{C}^{N_\mathrm{R}}$ represents the temporarily and spatially  additive white Gaussian noise vector.

	\vspace{-12pt}
	\subsection{Radar Model}
	The aim of the radar sensing task is to achieve the highest SNR toward target directions. The beampattern of the radar for $\Phi \in [-\frac{\pi}{2},\frac{\pi}{2}]$ is 
	\begin{align}
		B^{(i)}_m(\Phi) = \mathbf{a}_\mathrm{T}(\Phi)^\textsf{H}(\Phi)\mathbf{R}_\mathbf{x}^{(i)}[m] \mathbf{a}_\mathrm{T}(\Phi)  ,\label{beamPattern}
	\end{align}
	where $\mathbf{a}_\mathrm{T}(\Phi)\in\mathbb{C}^{N_\mathrm{T}}$ denotes the steering vector corresponding arbitrary target direction $\Phi$, and  $\mathbf{R}_\mathbf{x}^{(i)}[m]\in \mathbb{C}^{N_\mathrm{T}\times N_\mathrm{T}}$ is the covariance of the transmit signal. For the $i$-th spatial pattern, we have 
	\begin{align}
		\mathbf{R}_\mathbf{x}^{(i)}[m] &= \mathbb{E}\{ \mathbf{x}^{(i)}[m]\mathbf{x}^{(i)^\textsf{H}}[m] \}
		\nonumber\\	&
		=\mathbb{E}\{ \mathbf{F}_\mathrm{RF}^{(i)}\mathbf{F}_\mathrm{BB}^{(i)}[m]\mathbf{s}[m] \mathbf{s}^\textsf{H}[m] \mathbf{F}_\mathrm{BB}^{(i)^\textsf{H}}[m]\mathbf{F}_\mathrm{RF}^{(i)^\textsf{H}} \} \nonumber \\
		& 
		= \frac{1}{N_\mathrm{S}}\mathbf{F}_\mathrm{RF}^{(i)}\mathbf{F}_\mathrm{BB}^{(i)}[m] \mathbf{F}_\mathrm{BB}^{(i)^\textsf{H}}[m]\mathbf{F}_\mathrm{RF}^{(i)^\textsf{H}}.
	\end{align}
	{Then, we formulate the radar beampattern design problem as~\cite{beampatternDesign_Liu2018Feb,jrc_TCOM_Liu2020Feb}
		\begin{align}
			&\minimize_{\mathbf{F}_\mathrm{RF}^{(i)}, \{\mathbf{F}_\mathrm{BB}^{(i)}[m]\}_{m = 1}^M }  \sum_{m = 1}^M \sum_{\tilde{k}}^{\tilde{K} }  |\overline{ B}_m^{(i)}(\Phi_{\tilde{k}})  - B_m^{(i)}(\Phi_{\tilde{k}})    | \nonumber \\
			& \subjectto [\mathbf{R}_\mathbf{x}^{(i)}[m]]_{n,n} = 1/N_\mathrm{T}, \nonumber\\
			&\hspace{50pt} \mathbf{R}_\mathbf{x}^{(i)}[m] \succeq \mathbf{0}, \mathbf{R}_\mathbf{x}^{(i)}[m] = \mathbf{R}_\mathbf{x}^{(i)^\textsf{H}}[m],
		\end{align} 
		where $\overline{ B}_m^{(i)}(\Phi_{\tilde{k}})$ denotes the desired beampattern gain at direction $\Phi_{\tilde{k}}$.} 	To simultaneously obtain the desired beampattern for the radar target and achieve satisfactory communications performance, the hybrid beamformer $\mathbf{F}_\mathrm{RF}^{(i)}\mathbf{F}_\mathrm{BB}^{(i)}[m]$ should be designed accordingly.

	
	\section{Problem Formulation}
	\label{sec:ProbFormulation}
	While designing SPIM-ISAC hybrid beamformers, our goal is to maximize the {SE, which is characterized by mutual information (MI)~\cite{spim_GBMM_Guo2019Jul,spim_BIM_TVT_Ding2018Mar}.} In the sequel, we introduce the SE for both SPIM-assisted and conventional systems.
	\subsection{SE of the SPIM-ISAC System}
	Define $\mathrm{SE}(\mathbf{y}^{(i)}[m]; \mathbf{x}^{(i)}[m], \mathbf{F}^{(i)}[m])$ as the {SE} of the overall  wireless transmission for the received and transmitted signals $\mathbf{y}^{(i)}[m]$ and $\mathbf{x}^{(i)}[m]$ at the $i$-th spatial pattern with the hybrid beamformer $\mathbf{F}^{(i)}[m] = \mathbf{F}_\mathrm{RF}^{(i)}\mathbf{F}_\mathrm{BB}^{(i)}[m]$. Then, $\mathrm{SE}(\mathbf{y}^{(i)}[m]; \mathbf{x}^{(i)}[m], \mathbf{F}^{(i)}[m])$ is 
	\begin{align}
		\label{seOverall}
		&\mathrm{SE}(\mathbf{y}^{(i)}[m]; \mathbf{x}^{(i)}[m], \mathbf{F}^{(i)}[m]) = \mathrm{SE}(\mathbf{y}^{(i)}[m]; \mathbf{x}^{(i)}[m]|\mathbf{F}^{(i)}[m]  ) 
		\nonumber \\
		&\hspace{70pt}
		+ \mathrm{SE}(\mathbf{y}^{(i)}[m]; \mathbf{F}^{(i)}[m]), \hspace{10pt}m\in \mathcal{M},
	\end{align}
	where $\mathrm{SE}(\mathbf{y}^{(i)}[m]; \mathbf{x}^{(i)}[m]|\mathbf{F}^{(i)}[m]  )$ and $\mathrm{SE}(\mathbf{y}^{(i)}[m]; \mathbf{F}^{(i)}[m])$ stand for the SE corresponding to conventional symbol transmission and the SE achieved by employing SPIM, respectively. In particular, $\mathrm{SE}(\mathbf{y}^{(i)}[m]; \mathbf{x}^{(i)}[m]|\mathbf{F}^{(i)}[m]  )$ is well-known~\cite{heath2016overview} as
	\begin{align}
		&\mathrm{SE}(\mathbf{y}^{(i)}[m]; \mathbf{x}^{(i)}[m]|\mathbf{F}^{(i)}[m]  )\hspace{-3pt} = \hspace{-3pt} \frac{1}{S} \sum_{i = 1}^{S}\log_2 \mathrm{det}\{\boldsymbol{\Sigma}_{i}[m] \}, \label{MI_mmwave_only}
	\end{align}
	where $\boldsymbol{\Sigma}_{i}[m] = \mathbf{I}_{N_\mathrm{R}} +  \frac{1}{\sigma_n^2N_\mathrm{S}}\mathbf{H}[m]\mathbf{F}^{(i)}[m] {\mathbf{F}^{(i)}}^\textsf{H}[m]\mathbf{H}^\textsf{H}[m].$
	
	While there is no closed-form expression for $\mathrm{SE}(\mathbf{y}^{(i)}[m]; \mathbf{F}^{(i)}[m])$, it is lower-bounded by $\mathrm{SE}_\mathrm{LB}(\mathbf{y}^{(i)}[m]; \mathbf{F}^{(i)}[m])$~\cite{spim_bounds_JSTSP_Wang2019May}, which is defined as
	\begin{align}
		&\mathrm{SE}_\mathrm{LB}(\mathbf{y}^{(i)}[m]; \mathbf{F}^{(i)}[m]) = \log_2 S - N_\mathrm{R}\log_2 e
		\nonumber \\
		&\hspace{30pt} 
		- \frac{1}{S} \sum_{i = 1}^S \log_2  \left(\sum_{j = 1}^S \frac{\mathrm{det}\{ \boldsymbol{\Sigma}_i[m]\}}{ \mathrm{det}\{\boldsymbol{\Sigma}_{i}[m] + \boldsymbol{\Sigma}_{j}  [m]\} }\right), \label{MI_SPIM_only}
	\end{align} 
	{which has been shown to be a tight approximation of $\mathrm{SE}(\mathbf{y}^{(i)}[m]; \mathbf{F}^{(i)}[m])$~\cite{spim_GBMM_Guo2019Jul,spim_GBM_Gao2019Jul}. By combining $\mathrm{SE}(\mathbf{y}^{(i)}[m]; \mathbf{x}^{(i)}[m]|\mathbf{F}^{(i)}[m]  )$ in (\ref{MI_mmwave_only}) and $\mathrm{SE}_\mathrm{LB}(\mathbf{y}^{(i)}[m]; \mathbf{F}^{(i)}[m])$ in (\ref{MI_SPIM_only}),}  we obtain the SE of the SPIM-aided system as  			 
	\begin{align}
		\label{MI_SPIM}
		&\mathrm{SE}_\mathrm{SPIM}[m]  = \log_2 \left( \frac{S}{(2\sigma_n^2)^{N_\mathrm{R}}}\right)
		\nonumber\\
		&
		\hspace{20pt}
		- \frac{1}{S} \sum_{i = 1}^S \log_2  \left(\sum_{j = 1}^S \mathrm{det}\{\boldsymbol{\Sigma}_{i}[m] + \boldsymbol{\Sigma}_{j}[m]  \}^{-1}\right).
	\end{align}
	\subsection{SE of the MIMO-ISAC System}			
	In conventional MIMO  systems, the analog beamformer ${\mathbf{F}}_\mathrm{RF}$ relies on the selection of the strongest path for hybrid beamformer design~\cite{spim_GBMM_Guo2019Jul,spim_bounds_JSTSP_Wang2019May}. The SE expression is also the same for MIMO-ISAC. As an example, we have  $\mathbf{F}_\mathrm{RF}^{(1)} =[ \mathbf{F}_\mathrm{R},\mathbf{a}_\mathrm{T}(\theta_1)]$, where $\mathbf{F}_\mathrm{R}\in \mathbb{C}^{N_\mathrm{T}\times K}$ is the radar-only beamformer, and $\mathbf{a}_\mathrm{T}(\theta_1)$ corresponds to the strongest communications path with path gain ${\gamma}_{1}$. Since there is only one choice of transmission, i.e., $S=1$, the SE for the MIMO-ISAC system becomes
	\begin{align}
		&\mathrm{SE}_\mathrm{MIMO}[m]  =\log_2 \bigg(\mathrm{det}\bigg\{\mathbf{I}_{N_\mathrm{R}}
		\nonumber \\
		&\hspace{0pt}
		+ \frac{1}{\sigma_n^2N_\mathrm{S}}\mathbf{H}[m]{\mathbf{F}}_\mathrm{RF}^{(1)}{\mathbf{F}}_\mathrm{BB}^{(1)}[m]{\mathbf{F}_\mathrm{BB}^{(1)}}^\textsf{H}[m]{\mathbf{F}_\mathrm{RF}^{(1)}}^\textsf{H}\mathbf{H}^\textsf{H} [m]  \bigg \}  \bigg), \label{MI_mmwave1}
	\end{align}
	where $i=1$ denotes the first spatial pattern which, in this case, corresponds to the path with strongest gain. 	
	
	{
		The SE expressions in (\ref{MI_SPIM})  and  (\ref{MI_mmwave1}) are identical when there is only one spatial pattern (i.e., $S=1$). Furthermore, the SE in this scenario is maximized by the unconstrained beamformer $\mathbf{F}_\mathrm{opt}[m]\in \mathbb{C}^{N_\mathrm{T}\times N_\mathrm{S}}$, which results from the singular value decomposition (SVD) of $\mathbf{H}[m]$. Nevertheless, higher SE is achieved by employing multiple spatial patterns in (\ref{MI_SPIM}).
	}
	\subsection{Hybrid Beamformer Design}
	The SPIM-ISAC hybrid beamformer design problem is 
	\begin{subequations}
		\label{prob1}
		\begin{align}
			\maximize_{\mathbf{F}_\mathrm{RF}^{(i)}, \{\mathbf{F}_\mathrm{BB}^{(i)}[m]\}_{m = 1}^M} & \sum_{m = 1}^{M}\mathrm{SE}_\mathrm{SPIM}[m] \nonumber \\
			\subjectto \hspace{10pt}& \mathbf{F}_\mathrm{RF}^{(i)} \in  \mathcal{A},  \label{prob1_a}\\
			&|[\mathbf{F}_\mathrm{RF}^{(i)}]_{n,r}| = 1/\sqrt{N_\mathrm{T}},\label{prob1_b}\\
			& \sum_{m=1}^{M} \|\mathbf{F}_\mathrm{RF}^{(i)}\mathbf{F}_\mathrm{BB}^{(i)}[m]\|_\mathcal{F} = MN_\mathrm{S} ,\label{prob1_c} \\
			& {\sum_{m=1}^{M} \sum_{\tilde{k}}^{\tilde{K} }  |\overline{ B}_m^{(i)}(\Phi_{\tilde{k}})  - B_m^{(i)}(\Phi_{\tilde{k}})    | | \leq \rho ,  } \label{prob1_d}
		\end{align}
	\end{subequations}
	{where $\rho$ in (\ref{prob1_d}) is the sensing accuracy tolerance for beampattern design, and $\mathcal{A}= \{\mathbf{F}_\mathrm{RF}^{(1)},\cdots, \mathbf{F}_\mathrm{RF}^{(S)} \}$ in (\ref{prob1_a}) represents the set of possible analog beamformers for the SPIM.  Note that \eqref{prob1} also includes constraints for the constant-modulus property of $\mathbf{F}_\mathrm{RF}^{(i)}$  and the total power constraint as in (\ref{prob1_b}) and (\ref{prob1_c}), respectively.}
	
	The optimization problem in \eqref{prob1} belongs to the class of mixed-integer non-convex programming (MINCP)~\cite{mixedInteger_programming_Burer2012Jul}. It is computationally prohibitive because of the combinatorial subproblems for each spatial pattern $i$. It is also non-linear because of multiple unknowns $\mathbf{F}_\mathrm{RF}^{(i)}$ and $ \mathbf{F}_\mathrm{BB}^{(i)}[m]$. In order to provide an effective beamforming solution, we exploit the steering vectors corresponding to the radar and communications paths to design the beamformers for SPIM-ISAC. 

	{\section{ISAC Parameter Estimation}
		\label{sec:ParameterEst}
		
	}

	{In communications-only systems, the cost function in (\ref{prob1}) is maximized by designing the analog beamformer $\mathbf{F}_\mathrm{RF}^{(i)}$ from
		the steering vectors corresponding to the directions of the communications paths~\cite{spim_bounds_JSTSP_Wang2019May,spim_onGSM_He2017Sep}. In ISAC scenario, we utilize the radar and communications parameters, e.g., the directions of the radar targets and the communications paths. As a result, the proposed design ensures both communications and radar performance.} In particular, the analog beamformer is constructed from the steering vectors corresponding to the directions of the radar targets and communications user paths. Define the analog beamformer $\mathbf{F}_\mathrm{RF}^{(i)}\in \mathbb{C}^{N_\mathrm{T}\times N_\mathrm{RF}}$ as
	\begin{align}
		\label{Frf}
		\mathbf{F}_\mathrm{RF}^{(i)} = \left[\mathbf{F}_\mathrm{R} \;| \; \mathbf{F}_\mathrm{C}^{(i)} \right],
	\end{align}
	where $\mathbf{F}_\mathrm{R} =  [\mathbf{a}_\mathrm{T}(\Phi_1),\cdots, \mathbf{a}_\mathrm{T}(\Phi_K)]\in \mathbb{C}^{N_\mathrm{T}\times K}$ denotes the radar-only beamformer, which includes the steering vectors corresponding to the target directions. Also,  $\mathbf{F}_\mathrm{C}^{(i)}\in \mathbb{C}^{N_\mathrm{T}\times N_\mathrm{S}}$ is the communications-only analog beamformer comprised of the steering vectors corresponding to the communications paths for the $i$-th spatial pattern.
	
	{
		\begin{remark}
			We remark that the SPIM is only performed over the communications-only analog beamformers (i.e., $\mathbf{F}_\mathrm{C}^{(i)} $ in $	\mathbf{F}_\mathrm{RF}^{(i)} = \left[\mathbf{F}_\mathrm{R} \;| \; \mathbf{F}_\mathrm{C}^{(i)} \right]$). Therefore, the SPIM-ISAC system does not introduce additional degrees of freedom for sensing performance. 
		\end{remark}
	}
	

	\subsection{Radar Parameter Estimation}
	\label{sec:RadarParEst}
	The radar-only beamformer $\mathbf{F}_\mathrm{R}$ is constructed as the steering matrix corresponding to $\Phi_k$, $\forall k$, which are estimated during the search phase of the radar~\cite{elbir2021JointRadarComm}. To this end, the BS first transmits probing signals, which are reflected and processed by the BS to estimate the target directions.
	
	Define $\tilde{\mathbf{X}}_\mathrm{r}[m]\in \mathbb{C}^{N_\mathrm{T}\times T}$ as the radar probing signal  transmitted by the BS for $T$ data snapshots along the fast-time axis~\cite{mimoRadar_WidebandYu2019May,jrc_TCOM_Liu2020Feb}. {In particular, $\tilde{\mathbf{X}}_\mathrm{r}[m]\in \mathbb{C}^{N_\mathrm{T}\times T}$ is designed as OFDM linear frequency modulation (LFM) signal~\cite{lfm_ofdm_waveform_Wang2022Dec,jrc_TCOM_Liu2020Feb}. } $\tilde{\mathbf{X}}_\mathrm{r}[m]$ has the property $\mathbb{E}\{ \tilde{\mathbf{X}}_\mathrm{r}[m] \tilde{\mathbf{X}}_\mathrm{r}^\textsf{H}[m] \} = \frac{P_\mathrm{r}T}{M N_\mathrm{T}}\mathbf{I}_{N_\mathrm{T}}$, where $P_\mathrm{r}$ is the radar transmit power. The $N_\mathrm{RF}\times T$ echo signal reflected from the $K$ targets is
	\begin{align}
		\label{radarReceived}
		\tilde{\mathbf{Y}}[m] = \sum_{k = 1}^K {\beta}_k {e^{-\mathrm{j}2\pi \tilde{\tau}_k} } \tilde{\mathbf{a}}_\mathrm{T}(\Phi_k) \mathbf{a}_\mathrm{T}^\textsf{T}(\Phi_k) \tilde{\mathbf{X}}_\mathrm{r}[m] + \tilde{\mathbf{N}}[m],
	\end{align}
	{where ${\beta}_k\in \mathbb{C}$ and $\tilde{\tau}_k$ denote the reflection coefficient and the delay of the $k$-th target echo signal, respectively.} $\mathbf{a}_\mathrm{T}(\Phi_k)\in \mathbb{C}^{N_\mathrm{T}}$ is the transmit array steering vector corresponding to the $k$-th target DoA angle $\Phi_k$ and $\tilde{\mathbf{a}}_\mathrm{T}(\Phi_k) = \mathbf{W}_\mathrm{RF}^\textsf{H} {\mathbf{a}}_\mathrm{T}(\Phi_k)\in \mathbb{C}^{N_\mathrm{RF}}$ is the equivalent receive steering vector for $\mathbf{W}_\mathrm{RF}\in \mathbb{C}^{N_\mathrm{T}\times N_\mathrm{RF}}$ being the analog combiner matrix~\cite{jrc_TCOM_Liu2020Feb}, and $\tilde{\mathbf{N}}[m] = \mathbf{W}_\mathrm{RF}^\textsf{H}\bar{\mathbf{N}}[m]\in \mathbb{C}^{N_\mathrm{RF}\times T}$  represents the noise term, where $\bar{\mathbf{N}}[m] = [\bar{\mathbf{n}}_1[m],\cdots, \bar{\mathbf{n}}_T[m]]\in \mathbb{C}^{N_\mathrm{T}\times T}$ with $\bar{\mathbf{n}}_t[m]\sim \mathcal{CN}(\mathbf{0},\tilde{\sigma}_n^2\mathbf{I}_{N_\mathrm{T}})$. Denote the radar target steering matrix and reflection coefficients by $\tilde{\mathbf{A}}_\mathrm{T}(\Phi) = [\tilde{\mathbf{a}}_\mathrm{T}(\Phi_1),\cdots,\tilde{\mathbf{a}}_\mathrm{T}(\Phi_K)]$ and $\boldsymbol{\Xi} = \mathrm{diag}\{{\beta}_1e^{-\mathrm{j}2\pi \tilde{\tau}_1}, \cdots, {\beta}_Ke^{-\mathrm{j}2\pi \tilde{\tau}_K} \}\in \mathbb{C}^{K\times K}$, then (\ref{radarReceived}) becomes
	\begin{align}
		\label{arrayData}
		\tilde{\mathbf{Y}}[m] = \tilde{\mathbf{A}}_\mathrm{T}(\Phi)    \boldsymbol{\Xi} \mathbf{A}_\mathrm{T}^\textsf{T}(\Phi)\tilde{\mathbf{X}}_\mathrm{r}[m] + \tilde{\mathbf{N}}[m].
	\end{align}

	{In order to estimate the target directions, we invoke the wideband MUSIC algorithm~\cite{music,wideband_doaEst_Wideband_Friedlander1993Apr}.} Define $\mathbf{R}_{\tilde{\mathbf{Y}}}[m]\in \mathbb{C}^{N_\mathrm{RF}\times N_\mathrm{RF}}$ as the covariance matrix of $\tilde{\mathbf{Y}}[m]$, i.e., 
	\begin{align}
		\mathbf{R}_{\tilde{\mathbf{Y}}}[m] &= \frac{1}{T}\tilde{\mathbf{Y}}[m] \tilde{\mathbf{Y}}^\textsf{H}[m] \nonumber\\
		&= \frac{1}{T} \tilde{\mathbf{A}}_\mathrm{T}(\Phi) \left( \frac{P_\mathrm{r}T}{MN_\mathrm{T}}\widetilde{\boldsymbol{\Xi} }\right) \tilde{\mathbf{A}}_\mathrm{T}^\textsf{H}(\Phi) +  \frac{1}{T}\tilde{\mathbf{N}}[m]\tilde{\mathbf{N}}^\textsf{H}[m] \nonumber\\
		& \approx \frac{P_\mathrm{r}}{MN_\mathrm{T}} \tilde{\mathbf{A}}_\mathrm{T}(\Phi) \widetilde{\boldsymbol{\Xi} } \tilde{\mathbf{A}}_\mathrm{T}^\textsf{H}(\Phi)  +  \tilde{\sigma}_n^2 N_\mathrm{T} \mathbf{I}_{\mathrm{N}_\mathrm{RF}},
	\end{align}
	where $\tilde{\mathbf{N}}[m]\tilde{\mathbf{N}}^\textsf{H}[m] = \tilde{\sigma}_n^2 T \mathbf{W}_\mathrm{RF}^\textsf{H}\mathbf{W}_\mathrm{RF} \approx \tilde{\sigma}_n^2 TN_\mathrm{T}\mathbf{I}_{N_\mathrm{RF}} $ and  $\widetilde{\boldsymbol{\Xi} }\in \mathbb{C}^{K\times K} $ is defined as $	\widetilde{\boldsymbol{\Xi} } =  \boldsymbol{\Xi}\mathbf{A}_\mathrm{T}^\textsf{T}(\Phi)\mathbf{A}_\mathrm{T}^*(\Phi)\boldsymbol{\Xi}^*$. 	Then, the eigendecomposition of $\mathbf{R}_{\tilde{\mathbf{Y}}}[m]$ yields
	\begin{align}
		\label{covarianceY}
		\mathbf{R}_{\tilde{\mathbf{Y}}}[m] = \mathbf{U}[m] \boldsymbol{\Theta}[m] \mathbf{U}^\textsf{H}[m],
	\end{align}
	where $\boldsymbol{\Theta}[m]\in \mathbb{C}^{N_\mathrm{RF}\times N_\mathrm{RF}}$ is a diagonal matrix composed of the eigenvalues of $\mathbf{R}_{\tilde{\mathbf{Y}}}[m]$ in a descending order, and $\mathbf{U}[m] = \left[\mathbf{U}_\mathrm{S}[m]\hspace{2pt} \mathbf{U}_\mathrm{N}[m] \right]\in \mathbb{C}^{N_\mathrm{RF}\times N_\mathrm{RF}}$ corresponds to the eigenvector matrix; $\mathbf{U}_\mathrm{S}[m]\in\mathbb{C}^{N_\mathrm{RF}\times K}$ and $\mathbf{U}_\mathrm{N}[m]\in \mathbb{C}^{N_\mathrm{RF}\times N_\mathrm{RF}-K}$ are the signal and noise subspace eigenvector matrices, respectively. The columns of $\mathbf{U}_\mathrm{S}[m]$ and  $\tilde{\mathbf{A}}_\mathrm{T}(\Phi)$ span the same space that is orthogonal to the eigenvectors in $\mathbf{U}_\mathrm{N}[m]$ as 
	\begin{align}
		\label{musicCost}
		\| \mathbf{U}_\mathrm{N}^\textsf{H}[m]\tilde{\mathbf{a}}_\mathrm{T}(\Phi_k) \|_2^2 = 0,
	\end{align}
	for $k\in \mathcal{K}$ and $m\in \mathcal{M}$~\cite{music}. Thus, the estimates of the radar targets can be founds from the combined MUSIC spectra, i.e.,
	\begin{align}
		\label{musicSpectra2}
		\zeta(\Phi) = \sum_{m = 1}^M \zeta_{m}(\Phi),
	\end{align}
	where $	\zeta_{m}(\Phi) $ is the spectrum corresponding to the $m$-th subcarrier as $	\zeta_m(\Phi) = \frac{1}{\mathbf{a}_\mathrm{T}^\textsf{H}(\Phi)\mathbf{U}_\mathrm{N}[m]\mathbf{U}_\mathrm{N}^\textsf{H}[m] \mathbf{a}_\mathrm{T}(\Phi) }.$

	The MUSIC spectra in (\ref{musicSpectra2}) yields $MK$ peaks, which are deviated due to beam-split while correct MUSIC spectra should include $K$ peaks which are aligned for $m\in \mathcal{M}$. In other words, beam-split-corrected steering vectors should be used to accurately compute the MUSIC spectrum. This is the concept of our  BSA-MUSIC algorithm, described below, in which beam-split-corrected steering vectors are employed for the computation of the MUSIC spectrum. 
	
	
	Define $\mathbf{a}_\mathrm{T}(\Phi_m)\in\mathbb{C}^{N_\mathrm{T}}$ as the BSA SD steering vector for the nominal SI steering vector  $\mathbf{a}_\mathrm{T}(\Phi)$. The $n$-th entry of the BSA steering vector is explicitly defined as $[\mathbf{a}_\mathrm{T}(\Phi_m)]_n = \frac{1}{\sqrt{N_\mathrm{T}}}\exp \{- \mathrm{j} \pi (n-1) \Phi_m   \}$ whereas $[\mathbf{a}_\mathrm{T}(\Phi)]_n = \frac{1}{\sqrt{N_\mathrm{T}}}\exp \{- \mathrm{j} \frac{2\pi d}{\lambda_m } (n-1) \Phi   \}$. The beam-split correction implies that $\mathbf{a}_\mathrm{T}(\Phi) = \mathbf{a}_\mathrm{T}(\Phi_m)$ holds, such  that while the frequency $f_m$ varies, $\mathbf{a}_\mathrm{T}(\Phi_m)$ points to $\Phi$, whereas $\mathbf{a}_\mathrm{T}(\Phi)$ points to $\eta_m\Phi$. In other words, we have 
	\begin{align}
		{[\mathbf{a}_\mathrm{T}(\Phi_m)]_n} - {[\mathbf{a}_\mathrm{T}(\Phi)]_n} &= 0 \nonumber\\
		{\frac{1}{\sqrt{N_\mathrm{T}}}e^{-\mathrm{j} \pi(n-1)\Phi_m   }} - {\frac{1}{\sqrt{N_\mathrm{T}}}  e^{-\mathrm{j} \frac{2\pi \frac{\lambda_c}{2}  }{ \lambda_m }(n-1)\Phi   } } &= 0 \nonumber \\
		{e^{-\mathrm{j} \pi (n-1)\Phi_m   }} - { e^{-\mathrm{j} \pi \frac{\lambda_c}{\lambda_m} (n-1)\Phi   } }  &= 0, \label{def_a_am}
	\end{align}
	which yields $\Phi_m = \frac{\lambda_c}{\lambda_m}\Phi = \eta_m\Phi$.

	To provide further insight, we examine the array gain, which also holds for computing the MUSIC spectrum~\cite{music}, for wideband scenario in the following lemma, for which we define the array gain  $A_G(\Phi,m)$ for $\Phi$ at the $m$-th subcarrier as $	A_G(\Phi,m) = \frac{ |\mathbf{a}_\mathrm{T}^\textsf{H}(\Phi)\mathbf{a}_\mathrm{T}(\Phi_m)   |^2  }{N_\mathrm{T}^2}.$
	%

	\begin{lemma}
		\label{lemma1}
		Let $\mathbf{a}_\mathrm{T}(\Phi_m)$ and $\mathbf{a}_\mathrm{T}(\Phi)$ be the BSA and nominal steering vectors for an arbitrary direction $\Phi$ and subcarrier $m\in \mathcal{M}$ as defined in (\ref{def_a_am}), respectively. Then, $\mathbf{a}_\mathrm{T}(\Phi_m)$ achieves the maximum array gain, i.e., $A_G(\Phi,m) = \frac{ |\mathbf{a}_\mathrm{T}^\textsf{H}(\Phi)\mathbf{a}_\mathrm{T}(\Phi_m)   |^2  }{N_\mathrm{T}^2} $, if $\Phi_m = \eta_m \Phi $.
	\end{lemma}
	
	\begin{proof}
		See Appendix~\ref{appen2ArrayGain}.
	\end{proof}

	Using the aforementioned analysis and Lemma 1, the BSA-MUSIC spectrum is  $	\widetilde{\zeta}(\Phi) = \sum_{m = 1}^M \widetilde{\zeta}_{m}(\Phi),$ where 
	\begin{align}
		\label{musicSpectraBSA}
		\widetilde{\zeta}_{m}(\Phi) = \frac{1}{\boldsymbol{a}_m^\textsf{H}(\Phi)\mathbf{U}_\mathrm{N}[m]\mathbf{U}_\mathrm{N}^\textsf{H}[m] \boldsymbol{a}_m(\Phi) },
	\end{align}
	where $\boldsymbol{a}_m (\Phi) = \mathbf{W}_\mathrm{RF}^\textsf{H} \mathbf{a}_\mathrm{T}(\Phi_m) \in \mathbb{C}^{N_\mathrm{RF}}$ denotes the beam-split-corrected virtual steering vector. The $K$ highest peaks of the BSA-MUSIC spectrum in (\ref{musicSpectraBSA}) yields the radar target estimates $\{\hat{\Phi}_k\}_{k = 1}^K$.

	\begin{remark} Since there are limited number of RF chains, the size of the collected array data  in (\ref{radarReceived}) for the MUSIC algorithm is $N_\mathrm{RF}\times T$, which allows us to identify $K \leq  N_\mathrm{RF}-1$ targets. In order to improve the identifiability condition, full array data is collected via subarray processing. In other words, the array data is obtained at multiple time slots, say $T_\mathrm{slot} = N_\mathrm{T}/N_\mathrm{RF}$ provided that the phase alignment between the time slots is properly handled~\cite{widebandDoAEst_Hybrid_1_Shu2018Feb}. Then, the echo signal in (\ref{radarReceived}) is collected in $T_\mathrm{slot}$ time slots and the  $N_\mathrm{T}\times T$ array data is constructed as $		\widetilde{\mathbf{Y}}[m] = [\tilde{\mathbf{Y}}_{1}^\textsf{T}[m],\cdots, \tilde{\mathbf{Y}}_{T_\mathrm{slot}}^\textsf{T}[m]]^\textsf{T}\in \mathbb{C}^{N_\mathrm{T}\times T},$
		for which the identifiability condition is $K\leq N_\mathrm{T}-1$.
	\end{remark}

	%
	%
	
	\begin{figure}[t]
		\centering
		{\includegraphics[draft=false,width=.99\columnwidth]{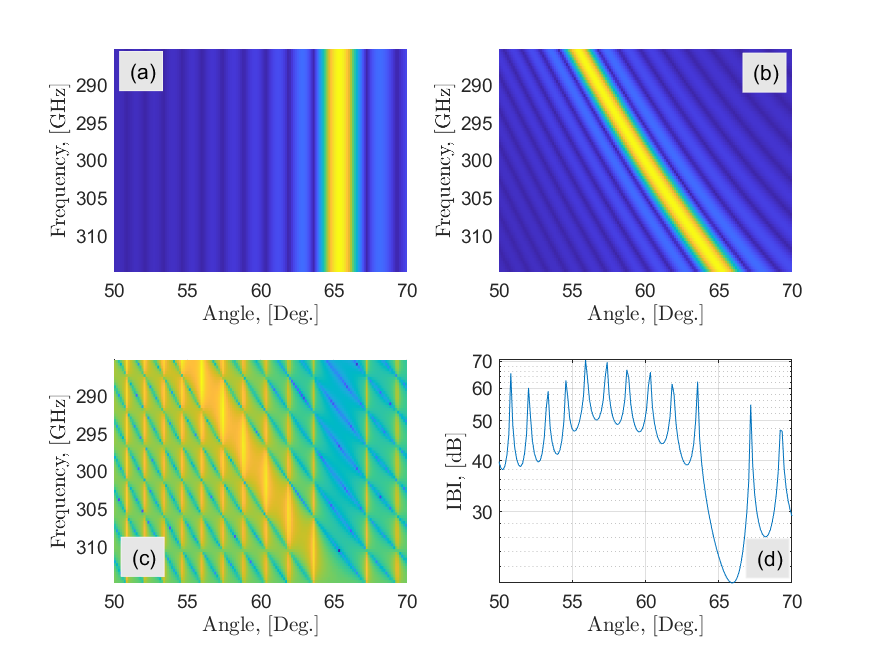} } 
		\caption{ {Array gain for (a) the beam-split-free and (b) beam-split-corrupted beams generated at $65^\circ$ and $60^\circ$, respectively. (c) The IBI computed for these two beams, and (d) the total IBI across the subcarriers.  }
		}
		
		\label{fig_IBI}
	\end{figure}

	{We further examine the inter-beam interference (IBI) in the presence of beam-split. The IBI from the beam direction $\Phi_i$ to $\Phi_j$ at frequency $f_m$ is~\cite{interBeamInterference1_Afeef2022Sep}
		\begin{align}
			\mathrm{IBI}(\Phi_i,\Phi_j,f_m) = \frac{A_G(\Phi_i,f_m)  }{A_G(\Phi_j,f_c  )},
		\end{align}	
		where $A_G(\Phi_i,f_m)$ and $A_G(\Phi_j,f_c)$ represent the beam-split-corrupted and beam-split-free array gains for the beam directions of $\Phi_i$ and $\Phi_j$, respectively. In order to illustrate the angular deviation across the subcarriers, we present the beam-split-free and beam-split-corrupted array gains of the beams generated at $65^\circ$ and $60^\circ$ in Fig.~\ref{fig_IBI}(a) and Fig.~\ref{fig_IBI}(b), respectively. Then, we show the IBI computed for these two beams in Fig.~\ref{fig_IBI}(c). The total IBI across the subcarriers is $\mathrm{IBI}(\Phi_i,\Phi_j) = \sum_{m=1}^{M} \mathrm{IBI}(\Phi_i,\Phi_j,f_m) $ (Fig.~\ref{fig_IBI}(b)). We observe that high IBI occurs for the angular sector of $[55^\circ, 65^\circ]$ because the beam-split-corrupted beam varies in this range with the change of subcarrier frequency.
	}

	\subsection{Communications Parameter Estimation}
	\label{sec:CommParEst}
	The communications-only analog beamformer ${\mathbf{F}}_\mathrm{C}^{(i)}$ is constructed from steering vectors corresponding to the path directions $\{\vartheta_l\}_{l = 1}^L$. 	This is achieved by the communication user feeding back $\{\vartheta_l\}_{l = 1}^L$ to the BS after the channel acquisition stage at the user side. 
	


	{In order to estimate the channel ${\mathbf{H}}[m]$, and eventually $\{\vartheta_l\}_{l = 1}^L$, we employ an OMP-based approach relying on a BSA dictionary. The key idea of the proposed BSA dictionary is to utilize the prior knowledge of $\eta_m$ to obtain beam-split-corrected steering vectors. Thus, a BSA dictionary is constructed, wherein the steering vectors are generated with the directions that are affected by beam-split. }Then, the physical direction  can readily be found  as $\phi = \frac{\varphi_{m}}{\eta_m}$, $\theta = \frac{\vartheta_{m}}{\eta_m}$ for an arbitrary spatial direction $\varphi_{m},\vartheta_{m}\in [-1,1]$, $\forall m\in \mathcal{M}$. Using this observation, we design the BSA dictionaries $\overline{\mathbf{P}}_m\in \mathbb{C}^{N_\mathrm{R}\times G}$ and $\overline{\mathbf{Q}}_m\in \mathbb{C}^{N_\mathrm{T}\times G}$, where $G$ is the grid size. Then, we have 
	\begin{align}
		\overline{\mathbf{P}}_m &= [{\mathbf{a}}_\mathrm{R}(\varphi_{1,m}),\cdots,{\mathbf{a}_\mathrm{R}}(\varphi_{G,m}) ],\\
		\overline{\mathbf{Q}}_m &= [\mathbf{a}_\mathrm{T}(\vartheta_{1,m}),\cdots,\mathbf{a}_\mathrm{T}(\vartheta_{G,m}) ], \label{bsadictionaries}
	\end{align}
	where  $\mathbf{a}_\mathrm{R}(\varphi_{g,m})$ and $\mathbf{a}_\mathrm{T}(\vartheta_{g,m})$ are $N_\mathrm{R}\times 1$ and $N_\mathrm{T}\times 1$ steering vectors for $g = 1,\cdots, G$. 
	


	
	In order to estimate the channel in downlink, the BS employs $J_\mathrm{T}$ beamformer vectors as $\tilde{\mathbf{F}}= [\tilde{\mathbf{f}}_1,\cdots, \tilde{\mathbf{f}}_{J_\mathrm{T}}]\in \mathbb{C}^{N_\mathrm{T}\times {J_\mathrm{T}}}$ to transmit ${J_\mathrm{T}}$ orthogonal pilots, $\tilde{\mathbf{S}}[m] = \mathrm{diag}\{\tilde{s}_1[m],\cdots, \tilde{s}_{J_\mathrm{T}}[m]\}\in \mathbb{C}^{{J_\mathrm{T}}\times {J_\mathrm{T}}}$. For the transmit pilots corresponding to each $\tilde{\mathbf{f}}_j$, the user with $\bar{N}_\mathrm{RF}$ RF chain employs  ${J_\mathrm{R}}$ ($J_\mathrm{R}\leq N_\mathrm{R}$) combining vectors $\tilde{\mathbf{w}}_j$ as $\tilde{\mathbf{W}} = [\tilde{\mathbf{w}}_1,\cdots, \tilde{\mathbf{w}}_{{J_\mathrm{R}}}]\in \mathbb{C}^{N_\mathrm{R}\times {J_\mathrm{R}}}$. Therefore, the total channel usage for processing all pilots during training is $J_\mathrm{T}\lceil\frac{{J_\mathrm{R}}}{\bar{{N}}_\mathrm{RF}}\rceil$. At the user side, the received  ${J_\mathrm{R}}\times {J_\mathrm{T}}$ signal is $	{\mathbf{Y}}[m] = \tilde{\mathbf{W}}^\textsf{H} \mathbf{H}[m]\tilde{\mathbf{F}}\tilde{\mathbf{S}} + \tilde{\mathbf{E}}[m],$
	where ${\mathbf{E}}[m] = \tilde{\mathbf{W}}^\textsf{H}\mathbf{N}[m]$ corresponds to the effective noise term. Assuming $\tilde{\mathbf{S}}[m] = \mathbf{I}_{J_\mathrm{T}}$, $\forall m\in \mathcal{M}$, we get 
	$	{\mathbf{Y}}[m] = \tilde{\mathbf{W}}^\textsf{H}\mathbf{H}[m]\tilde{\mathbf{F}}[m] + {\mathbf{E}}[m],$
	which is rewritten in vector form as
	\begin{align}
		\label{y_vector}
		\mathbf{y}[m] = ( \tilde{\mathbf{F}}^\textsf{T}\otimes \tilde{\mathbf{W}}^\textsf{H}) \mathbf{h}[m] + {\mathbf{e}}[m],
	\end{align}
	where $\mathbf{y}[m] = \mathrm{vec}\{{\mathbf{Y}}[m]\}\in \mathbb{C}^{{J_\mathrm{R}} {J_\mathrm{T}}}$,  $\mathbf{h}[m] = \mathrm{vec}\{\mathbf{H}[m]\}$ and ${\mathbf{e}}[m] = \mathrm{vec}\{{\mathbf{E}}[m]\}$. By exploiting the sparsity of the channel, (\ref{y_vector}) is rewritten as
	\begin{align}
		\mathbf{y}[m] = \boldsymbol{\Psi}_m  \mathbf{x}[m] + {\mathbf{e}}[m],\label{sparseSignalModel}
	\end{align}
	where $\mathbf{x}[m]\in\mathbb{C}^{G^2}$  is an $L$-sparse vector,
	and $\boldsymbol{\Psi}_m \in \mathbb{C}^{{J_\mathrm{R}}{J_\mathrm{T}}\times G^2}$ is the  dictionary matrix constructed from (\ref{bsadictionaries}) as 
	$	\boldsymbol{\Psi}_m = (\tilde{\mathbf{F}}^\textsf{T}\overline{\mathbf{P}}_m^*) \otimes (\tilde{\mathbf{W}}^\textsf{H} \overline{\mathbf{Q}}_m).$

	Given the received signal in (\ref{sparseSignalModel}), we  employ the OMP algorithm to effectively recover communications parameters $\{\hat{\phi}_l, \hat{\theta}_l, {\gamma}_l\}_{l = 1}^L$   by using the BSA-OMP approach presented in Algorithm~\ref{alg:BSACE}, wherein the physical path directions are found in steps $2$-$9$, and the channel  is reconstructed as $\hat{\overline{\mathbf{H}}}[m]$ from beam-split-corrected array responses $\hat{\mathbf{P}}$, $\hat{\mathbf{Q}}$ and $\hat{\boldsymbol{\Lambda}}_m$ in  steps $10$-$14$. Note that the complexity order of BSA-OMP is the same as that of conventional OMP techniques~\cite{heath2016overview}.

	{
		\section{Beamformer Design for SPIM-ISAC}
		\label{sec:SPIMISAC}
		Once the radar and communications parameters are estimated, the task is to design the hybrid beamformers for SPIM-ISAC system. Towards this end, we propose a two-step approach: obtain the analog beamformers using the estimated radar and communications parameters in Sec.~\ref{sec:RadarParEst} and Sec.\ref{sec:CommParEst}, respectively.
	}

	The analog beamformer $\mathbf{F}_\mathrm{RF}^{(i)}\in \mathbb{C}^{N_\mathrm{T}\times N_\mathrm{RF}}$ is comprised of the radar and communications analog beamformers, i.e., $	\mathbf{F}_\mathrm{RF}^{(i)} = \left[\mathbf{F}_\mathrm{R} \;| \; \mathbf{F}_\mathrm{C}^{(i)} \right]$ as in (\ref{Frf}),	where the radar-only analog beamformer $\mathbf{F}_\mathrm{R}\in \mathbb{C}^{N_\mathrm{T}\times K}$ is $	\mathbf{F}_\mathrm{R} = \left[\mathbf{a}_\mathrm{T}(\hat{\Phi}_1), \cdots, \mathbf{a}_\mathrm{T}(\hat{\Phi}_K)\right].$
	Similarly, the communication-only analog beamformer for the $i$-th spatial pattern, i.e., $\mathbf{F}_\mathrm{C}^{(i)} \in \mathbb{C}^{N_\mathrm{T}\times L_\mathrm{S}}$ is $	\mathbf{F}_\mathrm{C}^{(i)} = \left[ \mathbf{a}_\mathrm{T}(\hat{\theta}_1), \cdots, \mathbf{a}_\mathrm{T}(\hat{\theta}_L)  \right]\mathbf{B}^{(i)},$ 
	where $\mathbf{B}^{(i)} $ is an $L\times L_\mathrm{S}$ selection matrix selecting the steering vectors corresponding to the $L_\mathrm{S}$ out of $L$ spatial paths for the $i$-th spatial pattern with the structure of $	\mathbf{B}^{(i)} = \left[\mathbf{b}_{i_1}, \cdots, \mathbf{b}_{i_{L_\mathrm{S}}}  \right],$
		%
	where $\mathbf{b}_{i_l}$ is the $i_l$-th column of identity matrix $\mathbf{I}_{L_\mathrm{S}}$. {Since $\mathbf{b}_{i_l}$ always includes a non-zero entry, $\mathbf{F}_\mathrm{C}^{(i)} \neq \emptyset$, $\forall i\in \mathcal{S}$. In other words, even if the transmitted SPIM bits are all zero, there exists an $N_\mathrm{T}\times L_\mathrm{S}$ analog beamformer  $\mathbf{F}_\mathrm{C}^{(i)}$ with a certain spatial pattern.  }

	We consider following three approaches for beamforming, which are also summarized in Algorithm~\ref{alg:HB}.

	\begin{algorithm}[t]
		\begin{algorithmic}[1] 
			\caption{ \bf Communications parameter estimation}
			\color{black}
			\Statex {\textbf{Input:}    $\mathbf{y}[m]$, $\boldsymbol{\Psi}_m$ and $\eta_m$, $\forall m\in \mathcal{M}$. \label{alg:BSACE}}
			\State  $l=1$, $\bar{\mathcal{I}}_{l-1} = \mathcal{I}_{l-1} = \emptyset$,
			$\mathbf{r}_{l-1}[m] = \mathbf{y}[m], \forall m\in \mathcal{M}$.
			\State \textbf{while} $l\leq L$ \textbf{do}
			
			\State \indent $\{u^\star, v^\star \}= \argmax_{u,v} \sum_{m=1}^{M}|\boldsymbol{\psi}_{u,v}^\textsf{H}[m]\mathbf{r}_{l-1}[m] |$, \par  \indent where $\hspace{-1pt}\boldsymbol{\psi}_{u,v}[m] \hspace{-1pt}=\hspace{-1pt} (\tilde{\mathbf{F}}^\textsf{T}\mathbf{a}_\mathrm{R}^*(\varphi_{u,m})) \otimes (\tilde{\mathbf{W}}^\textsf{H} {\mathbf{a}}_\mathrm{T}(\vartheta_{v,m})) $.
			\State \indent $\bar{\mathcal{I}}_{l} \gets \bar{\mathcal{I}}_{l-1} \bigcup \{u^\star \}$, $\hat{\phi}_{l} =\frac{\varphi_{u^\star,m}}{\eta_m}$.
			\State \indent $\mathcal{I}_{l} \gets \mathcal{I}_{l-1} \bigcup \{v^\star\}$, $\hat{\theta}_{l} =\frac{\vartheta_{v^\star,m}}{\eta_m}$.
			\State  \indent $\boldsymbol{\Psi}_m(\bar{\mathcal{I}}_l,\mathcal{I}_l) = (\tilde{\mathbf{F}}^\textsf{T}\overline{\mathbf{Q}}_m(\mathcal{I}_l)) \otimes (\tilde{\mathbf{W}}^\textsf{H} \overline{\mathbf{P}}_m(\bar{\mathcal{I}}_l))$.
			\State  \indent $\mathbf{r}_{l}[m] = \left( \mathbf{I}_{J_\mathrm{R}J_\mathrm{T}} -  \boldsymbol{\Psi}_m(\bar{\mathcal{I}}_l,\mathcal{I}_l) \boldsymbol{\Psi}_m^\dagger(\bar{\mathcal{I}}_l,\mathcal{I}_l) \right) \mathbf{y}[m]$.
			\State \indent $l\gets l+ 1$.
			\State \textbf{end while}
			\State   $\hat{\mathbf{P}}=  [{\mathbf{a}}_\mathrm{R}(\hat{\phi}_{1}), \cdots, {\mathbf{a}_\mathrm{R}}(\hat{\phi}_{L})]$. $\hat{\mathbf{Q}} =  [\mathbf{a}_\mathrm{T}(\hat{\theta}_{1}), \cdots, \mathbf{a}_\mathrm{T}(\hat{\theta}_{L})]$.
			\State  \textbf{for} $m\in \mathcal{M}$
			\State \indent $\mathrm{diag}\{\hat{\boldsymbol{\Lambda}}_m\} = \boldsymbol{\Psi}_m^\dagger(\bar{\mathcal{I}}_L,\mathcal{I}_{L}) \mathbf{y}[m]$.
			\State \indent$\hat{\overline{\mathbf{H}}}[m] = \hat{\mathbf{P}}\hat{\boldsymbol{\Lambda}}_m\hat{\mathbf{Q}}^\textsf{H}$.
			\State \textbf{end for}
			\Statex \textbf{Output:} $\hat{\mathbf{P}}$, $\hat{\mathbf{Q}}$ and $\hat{\boldsymbol{\Lambda}}_m$, $\hat{\overline{\mathbf{H}}}[m]$.
		\end{algorithmic} 
	\end{algorithm}

	\vspace{-12pt}
	\subsection{Hybrid Beamforming}
	\vspace{-6pt}
	Given the analog beamformer $\mathbf{F}_\mathrm{RF}^{(i)}$ as in (\ref{Frf}), the baseband beamformer is computed by minimizing the Euclidean distance between the hybrid beamformer and the joint radar-communications (JRC) beamformer, which is defined as $\mathbf{F}_\mathrm{CR}[m]\in \mathbb{C}^{N_\mathrm{T}\times N_\mathrm{S}}$. Specifically, $\mathbf{F}_\mathrm{CR}[m]$ is composed of radar-only beamformer $\mathbf{F}_\mathrm{R}$ and the unconstrained communication-only beamformer $\mathbf{F}_\mathrm{opt}[m]\in\mathbb{C}^{N_\mathrm{T}\times N_\mathrm{S}}$ (which can be obtained through the singular value decomposition (SVD) of $\mathbf{H}[m]$ (i.e., the singular vectors corresponding to the $N_\mathrm{S}$ largest singular values of $\mathbf{H}[m]$)~\cite{heath2016overview}).  Then, the JRC beamformer is defined as
	\begin{align}
		\label{Fcr}
		\mathbf{F}_\mathrm{CR}[m] = \varepsilon \mathbf{F}_\mathrm{opt}[m] + (1- \varepsilon) \mathbf{F}_\mathrm{R}\boldsymbol{\Pi}[m],
	\end{align} 
	where  $\boldsymbol{\Pi}[m]\in\mathbb{C}^{K\times N_\mathrm{S}}$ is a unitary matrix providing the change of dimensions between $\mathbf{F}_\mathrm{R}$ and $\mathbf{F}_\mathrm{opt}[m]$. {In (\ref{Fcr}), $0\leq \varepsilon\leq 1$ represents  the trade-off parameter between the radar and communications tasks. In particular, $\varepsilon=1$ ($\varepsilon = 0$) corresponds to the communications-only (radar-only) design. In ISAC, $\varepsilon$ controls	the trade-off between the accuracy/prominence of sensing and communications tasks~\cite{elbir_thz_jrc_Magazine_Elbir2022Aug}. The selection procedure of  $\varepsilon$ in the relevant literature includes the ratio of power budgets~\cite{tradeoff_parameterSelection_Chiriyath2015Sep} and the signal durations percentages of the coherent processing interval~\cite{tradeoff_CPI_Dokhanchi2019Feb} allocated for radar and communications tasks.} {Different beamformer designs may also be employed by tuning the trade-off parameter $\varepsilon$. For instance, the transmitted waveform can be optimized given the desired radar beampattern and the desired constellation symbol matrix~\cite{fanLiu_waveformDesign_Liu2018Jun}. However, our formulation in (\ref{Fcr}) provides a simple architecture to design the JRC hybrid beamformer that is obtained via conventional optimization techniques applied to communications-only systems~\cite{hybridBFAltMin,mimoRFChainHybrid}.}
	
	{Given   the JRC beamformer, the hybrid beamformer design problem is
		\begin{align}
			\label{problemJRC1}
			\minimize_{\mathbf{F}_\mathrm{RF}^{(i)},\{  \mathbf{F}_\mathrm{BB}^{(i)}[m]\}_{m =1}^M}  &\hspace{3pt} \frac{1}{M}\sum_{m =1}^M\|\mathbf{F}_\mathrm{RF}^{(i)}\mathbf{F}_\mathrm{BB}^{(i)}[m]  -  \mathbf{F}_\mathrm{CR}[m]\|_\mathcal{F} \nonumber \\
			\subjectto & \sum_{m =1}^M\| \mathbf{F}_\mathrm{RF}^{(i)}\mathbf{F}_\mathrm{BB}^{(i)}[m] \|_\mathcal{F} = MN_\mathrm{S}, \nonumber \\
			& |[\mathbf{F}_\mathrm{RF}^{(i)}]_{n,r}| = {1}/{\sqrt{N_\mathrm{T}}}.
		\end{align}
		This formulation ignores SPIM because $ \mathbf{F}_\mathrm{CR}[m]$ is not exclusively defined for a spatial pattern. Therefore, we follow an alternating approach to design the hybrid beamformers. First, we construct the analog beamformer and follow it with optimizing the baseband beamformer $\mathbf{F}_\mathrm{BB}^{(i)}[m]$ and the auxiliary matrix $\boldsymbol{\Pi}[m]$. 
	}
	
	By using the analog beamformer $\mathbf{F}_\mathrm{RF}^{(i)}$ and  $\mathbf{F}_\mathrm{CR}[m]$, the baseband beamformer corresponding to the $i$-th spatial pattern is 
	\begin{align}
		\mathbf{F}_\mathrm{BB}^{(i)}[m] = {\mathbf{F}_\mathrm{RF}^{(i)}}^\dagger \mathbf{F}_\mathrm{CR}[m],
	\end{align}
	which is then normalized as $\mathbf{F}_\mathrm{BB}^{(i)}[m] = \frac{\sqrt{N_\mathrm{S}} {\mathbf{F}_\mathrm{RF}^{(i)}}^\dagger \mathbf{F}_\mathrm{CR}[m]  }{\|\mathbf{F}_\mathrm{RF}^{(i)}\mathbf{F}_\mathrm{BB}^{(i)}[m]    \|_\mathcal{F}}$. The JRC beamformer is composed of the auxiliary matrix $\boldsymbol{\Pi}[m]$, which can be optimized as 
	\begin{align}
		\label{prob_FBB_P}
		&\minimize_{\overline{\boldsymbol{\Pi}}} \hspace{3pt} \|\mathbf{F}_\mathrm{RF}^{(i)}\overline{\mathbf{F}}_\mathrm{BB}^{(i)} - \overline{\mathbf{F}}_\mathrm{CR}   \|_\mathcal{F}^2  \nonumber \\
		&	\subjectto \hspace{5pt} \overline{\boldsymbol{\Pi}}\; \overline{\boldsymbol{\Pi}}^\textsf{H} = \mathbf{I}_K,
	\end{align}
	where $\overline{\mathbf{F}}_\mathrm{BB}^{(i)} = \left[ \mathbf{F}_\mathrm{BB}^{(i)}[1], \cdots, \mathbf{F}_\mathrm{BB}^{(i)}[M] \right]$, $\overline{\mathbf{F}}_\mathrm{CR} = \left[ \mathbf{F}_\mathrm{CR}[1], \cdots, \mathbf{F}_\mathrm{CR}[M] \right]$ and $\overline{\boldsymbol{\Pi}} = \left[\boldsymbol{\Pi}[1],\cdots, \boldsymbol{\Pi}[M] \right]$ are $N_\mathrm{RF}\times MN_\mathrm{S}$, $N_\mathrm{T}\times MN_\mathrm{S}$ and $K\times MN_\mathrm{S}$ matrices composed of information corresponding to all subcarriers, respectively. 
	
	The problem in (\ref{prob_FBB_P}) is called orthogonal Procrustes problem (OPP), and its solution can be found via SVD of the $K\times MN_\mathrm{S}$ matrix $\mathbf{F}_\mathrm{R}^\textsf{H} \mathbf{F}_\mathrm{RF}^{(i)} \overline{\mathbf{F}}_\mathrm{BB}^{(i)}$ and it is given by~\cite{procrustesProblem_Hurley1962Apr} $	\overline{\boldsymbol{\Pi}} = \widetilde{\boldsymbol{\Pi}} \mathbf{I}_{K\times MN_\mathrm{S}} \widetilde{\mathbf{V}},$
	where $\widetilde{\boldsymbol{\Pi}} \widetilde{\boldsymbol{\Sigma}} \widetilde{\mathbf{V}}  =  \mathbf{F}_\mathrm{R}^\textsf{H} \mathbf{F}_\mathrm{RF}^{(i)} \overline{\mathbf{F}}_\mathrm{BB}^{(i)}$ is the SVD of the $N_\mathrm{RF}\times N_\mathrm{S}$ matrix $\frac{1}{1- \varepsilon }\mathbf{F}_\mathrm{R}^\textsf{H} \left(\mathbf{F}_\mathrm{RF}^{(i)} \overline{\mathbf{F}}_\mathrm{BB}^{(i)} -  \varepsilon \overline{\mathbf{F}}_\mathrm{CR}\right)$, and $\mathbf{I}_{K\times MN_\mathrm{S}}  = \left[\mathbf{I}_K \hspace{1pt}|  \hspace{1pt} \mathbf{0}_{ MN_\mathrm{S}- K\times K}^\textsf{T}  \right]^\textsf{T}$. Then, by estimating $\mathbf{F}_\mathrm{BB}^{(i)}[m]$ and $\boldsymbol{\Pi}[m]$ iteratively, the hybrid beamformer weights are computed.

	\subsection{BSA Hybrid Beamforming}
	\label{sec:BeamSplitMitigation}
	
	As discussed in Section~\ref{sec:beamSplit}, beam-split can be compensated if SD analog beamformers are used. However, this approach is costly since it requires employing $MN_\mathrm{T}N_\mathrm{RF}$ (instead of $N_\mathrm{T}N_\mathrm{RF}$) phase-shifters. Instead, we propose an efficient BSA approach, wherein the effect of beam-split is handled in the baseband beamformer, which is SD. Therefore, the effect of beam-split is conveyed from analog domain to baseband.

	Denoted by $\breve{\mathbf{F}}_\mathrm{RF}^{(i)}[m]\in \mathbb{C}^{N_\mathrm{T}\times N_\mathrm{RF}}$,  the SD analog beamformer that can be computed from the SI analog beamformer $\mathbf{F}_\mathrm{RF}^{(i)}$ as $	\breve{\mathbf{F}}_\mathrm{RF}^{(i)}[m] = \frac{1}{\sqrt{N_\mathrm{T}}}   \boldsymbol{\Omega}^{(i)}[m],$
	where $\boldsymbol{\Omega}^{(i)}[m]\in \mathbb{C}^{N_\mathrm{T}\times N_\mathrm{RF}}$ includes the angle information of $\mathbf{F}_\mathrm{RF}^{(i)}$ as $[\boldsymbol{\Omega}^{(i)}[m]]_{n,j} = \exp \{\mathrm{j} \eta_m \angle \{[\mathbf{F}_\mathrm{RF}^{(i)}]_{n,j} \}\}$ for $n = 1,\cdots, N_\mathrm{T}$ and $j =1,\cdots, N_\mathrm{RF}$. As a result, the angular deviation in $\mathbf{F}_\mathrm{RF}^{(i)}$ due to beam-split is compensated  with $\eta_m$. 
	
	Now, we define $\widetilde{\mathbf{F}}_\mathrm{BB}^{(i)}[m]\in \mathbb{C}^{N_\mathrm{RF}\times N_\mathrm{S}}$ as the \textit{BSA digital beamformer} in order to achieve SD beamforming performance that can be obtained by the usage of SD analog beamformer $\breve{\mathbf{F}}_\mathrm{RF}^{(i)}[m]$. Hence, we aim to match the proposed \textit{BSA hybrid beamformer} $\mathbf{F}_\mathrm{RF}^{(i)} \widetilde{\mathbf{F}}_\mathrm{BB}^{(i)}[m]$ with the SD hybrid beamformer $\breve{\mathbf{F}}_\mathrm{RF}^{(i)}[m] \mathbf{F}_\mathrm{BB}^{(i)}[m] $ as
	\begin{align}
		\minimize_{\widetilde{\mathbf{F}}_\mathrm{BB}^{(i)}[m]} \| \mathbf{F}_\mathrm{RF}^{(i)} \widetilde{\mathbf{F}}_\mathrm{BB}^{(i)}[m] - \breve{\mathbf{F}}_\mathrm{RF}^{(i)}[m] \mathbf{F}_\mathrm{BB}^{(i)}[m] \|_\mathcal{F}^2,
	\end{align}
	for which $\widetilde{\mathbf{F}}_\mathrm{BB}^{(i)}[m]$ can be obtained as
	\begin{align}
		\label{fbbTilde}
		\widetilde{\mathbf{F}}_\mathrm{BB}^{(i)}[m] = {\mathbf{F}_\mathrm{RF}^{(i)}}^\dagger \breve{\mathbf{F}}_\mathrm{RF}^{(i)}[m] \mathbf{F}_\mathrm{BB}^{(i)}[m].
	\end{align}
	\begin{remark} Because of the reduced dimension of the baseband beamformer (i.e., $N_\mathrm{RF}< N_\mathrm{T}$), the BSA approach does not completely mitigate beam-split. Instead,  the beam-split is fully mitigated only if  ${\mathbf{F}_\mathrm{RF}^{(i)}}^\dagger \breve{\mathbf{F}}_\mathrm{RF}^{(i)}[m]  = \mathbf{I}_{\mathrm{N}_\mathrm{T}}$ so that the resulting hybrid beamformer $\mathbf{F}_\mathrm{RF}^{(i)} \widetilde{\mathbf{F}}_\mathrm{BB}^{(i)}[m]$ is equal to  $\breve{\mathbf{F}}_\mathrm{RF}^{(i)}[m] \mathbf{F}_\mathrm{BB}^{(i)}[m]$,  which requires $N_\mathrm{RF} = N_\mathrm{T}$. Nevertheless, the proposed approach yields satisfactory SE performance with beam-split compensation for a wide range of bandwidths (see Fig.~\ref{fig_SE_BW}).
	\end{remark}

	\subsection{SI- and SD-AO Beamforming}
	\label{sec:AOBF}
	The proposed $N_\mathrm{T}\times N_\mathrm{S}$ SI-AO beamformer is given by $\mathbf{F}_\mathrm{SI-AO}^{(i)} = \mathbf{F}_\mathrm{RF}^{(i)} \mathbf{D}$, where $\mathbf{D}\in \mathbb{C}^{N_\mathrm{RF}\times N_\mathrm{S}}$ is an amplitude controller matrix as $	\mathbf{D} = \left[ \begin{array}{c}
		(1-\varepsilon) \mathbf{I}_{K\times N_\mathrm{S}} \\
		\varepsilon \mathbf{I}_{N_\mathrm{S}}	\end{array}    \right],$
	which allows the trade-off between the radar and communication tasks~\cite{elbir2021JointRadarComm}, and it can be realized via variable gain amplifiers~\cite{amplitudeControl_Lee2020May}. Despite its simple structure, the AO baseband beamformer can demonstrate satisfactory SE performance (see Sec.~\ref{sec:Sim}).
	The SD-AO beamformer has a similar structure, but it employs SD analog beamformer as $\mathbf{F}_\mathrm{SD-AO}^{(i)}[m] = \breve{\mathbf{F}}_\mathrm{RF}^{(i)}[m] \mathbf{D}$, which, therefore, employs $MN_\mathrm{T}N_\mathrm{RF}$ phase shifters.


		\begin{algorithm}[t]
			\begin{algorithmic}[1]
				\caption{ \bf Hybrid beamformer design}
				\color{black}
				\Statex {\textbf{Input:}   $\{\hat{\Phi}_k \}_{k = 1}^K$,  $\{\hat{\phi}_l, \hat{\theta}_l, \hat{\gamma}_l\}_{l = 1}^L$, $i\in \mathcal{S}$, $\eta_m$ and $\tilde{\epsilon}$. \label{alg:HB}}
				\State $\mathbf{F}_\mathrm{R} = \left[\mathbf{a}(\hat{\Phi}_1), \cdots, \mathbf{a}(\hat{\Phi}_K)\right].$
				\State $\mathbf{F}_\mathrm{C}^{(i)} = \left[ \mathbf{a}_\mathrm{T}(\hat{\theta}_1), \cdots, \mathbf{a}_\mathrm{T}(\hat{\theta}_L)  \right]\mathbf{B}^{(i)}.$
				\State $	\mathbf{F}_\mathrm{RF}^{(i)} = \left[\mathbf{F}_\mathrm{R} \;| \; \mathbf{F}_\mathrm{C}^{(i)} \right].$
				\Statex \textbf{$\star$ Hybrid beamformer:} 
				\State \textbf{while} $\epsilon < \tilde{\epsilon}$ \textbf{do}
				\State \indent $	\mathbf{F}_\mathrm{CR}[m] = \varepsilon \mathbf{F}_\mathrm{opt}[m] + (1- \varepsilon) \mathbf{F}_\mathrm{R}\boldsymbol{\Pi}[m]$, $m\in \mathcal{M}$.
				\State \indent $\mathbf{F}_\mathrm{BB}^{(i)}[m] = {\mathbf{F}_\mathrm{RF}^{(i)}}^\dagger \mathbf{F}_\mathrm{CR}[m]$, $m\in \mathcal{M}$.
				\State \indent $\overline{\boldsymbol{\Pi}} = \widetilde{\boldsymbol{\Pi}} \mathbf{I}_{K\times MN_\mathrm{S}} \widetilde{\mathbf{V}}$.
				\State \indent $\epsilon = \sum_{m = 1}^{M}\|\mathbf{F}_\mathrm{RF}^{(i)}\mathbf{F}_\mathrm{BB}^{(i)}[m] - \mathbf{F}_\mathrm{CR}[m]   \|_\mathcal{F}^2$.
				\State \textbf{end}
				\Statex  \textbf{$\star$ BSA hybrid beamformer:}
				\State  $\breve{\mathbf{F}}_\mathrm{RF}^{(i)}[m] = \frac{1}{\sqrt{N_\mathrm{T}}}   \boldsymbol{\Omega}^{(i)}[m]$ where $[\boldsymbol{\Omega}^{(i)}[m]]_{n,j} = \exp \{\mathrm{j} {\eta_m} \angle \{[\mathbf{F}_\mathrm{RF}^{(i)}]_{n,j} \}\}$.
				\State   $\widetilde{\mathbf{F}}_\mathrm{BB}^{(i)}[m] = {\mathbf{F}_\mathrm{RF}^{(i)}}^\dagger \breve{\mathbf{F}}_\mathrm{RF}^{(i)}[m] \mathbf{F}_\mathrm{BB}^{(i)}[m]$.
				\Statex  \textbf{$\star$ SI-AO beamformer:} $\mathbf{F}_\mathrm{SI-AO}^{(i)} = \mathbf{F}_\mathrm{RF}^{(i)} \mathbf{D}$.
				\Statex \textbf{$\star$ SD-AO beamformer:} $\mathbf{F}_\mathrm{SD-AO}^{(i)}[m] = \breve{\mathbf{F}}_\mathrm{RF}^{(i)}[m] \mathbf{D}$.

				\Statex \textbf{Output:} $\mathbf{F}_\mathrm{RF}^{(i)}$, $\mathbf{F}_\mathrm{BB}^{(i)}[m]$, $\mathbf{F}_\mathrm{SI-AO}^{(i)}$, $\mathbf{F}_\mathrm{SD-AO}^{(i)}[m]$,  $\widetilde{\mathbf{F}}_\mathrm{BB}^{(i)}[m]$. 
			\end{algorithmic} 
		\end{algorithm}

		\begin{figure}[t]
			\centering
			{\includegraphics[draft=false,width=.99\columnwidth]{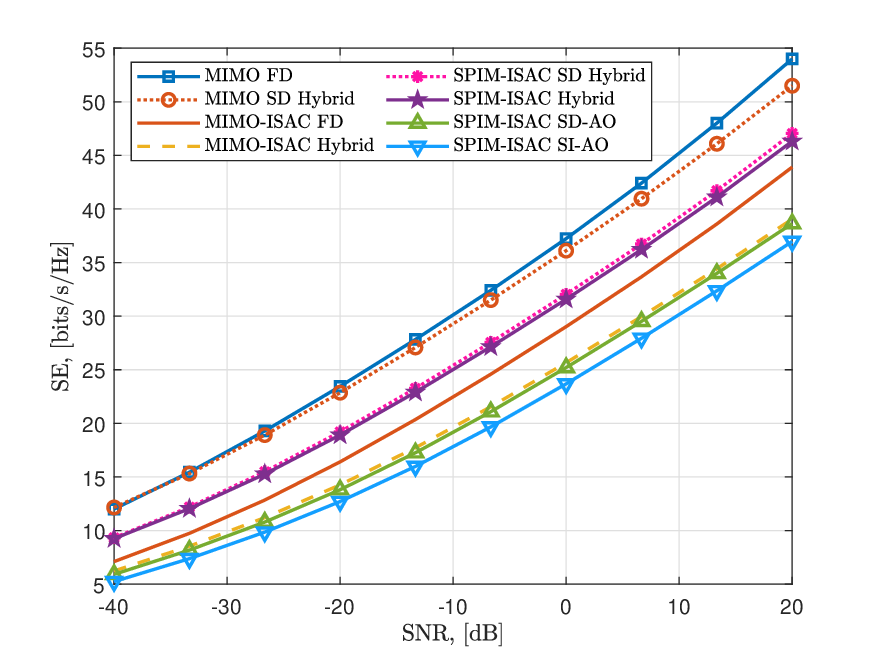} } 
			\caption{SE versus SNR when the radar-communications trade-off parameter $\varepsilon = 0.5$.
			}
			
			\label{fig_SE_SNR}
		\end{figure}


		\section{Numerical Experiments}
		\label{sec:Sim}
		
		We evaluated the performance of our SPIM-ISAC approach in comparison with FD and hybrid beamforming for MIMO-ISAC as well as SPIM-ISAC with SD-AO (Sec.~\ref{sec:AOBF}) and SI-AO beamformers~\cite{spim_bounds_JSTSP_Wang2019May}, in terms of SE and beamforming gain averaged over $500$ Monte Carlo trials. The number of antennas at the BS and the users are $N_\mathrm{T}=128$ and $N_\mathrm{R}=16$, respectively. The  carrier frequency and the bandwidth are selected as $f_c$ and $B = \frac{f_c}{10}$, respectively. {The number of subcarriers is $M=64$ and the grid size is set to $G=8N_\mathrm{T}$.} We select the number of available spatial paths, unless stated otherwise, as $L=8$ ($L_\mathrm{S} = 3$) and the number of targets is $K=2$. Thus, $P = 10$, $N_\mathrm{RF} = 5$ and $ N_\mathrm{S} = 3$. The target and path directions are drawn from $[-90^\circ,90^\circ]$ uniformly at random, while the path gains are selected as $\gamma_l \sim \mathcal{N}(1, (0.1)^2)$, $\forall l$~\cite{delayPhasePrecoding_THz_Dai2022Mar,elbir2021JointRadarComm}.

		\begin{figure}[t]
			\centering
			\subfloat[]{\includegraphics[draft=false,width=\columnwidth]{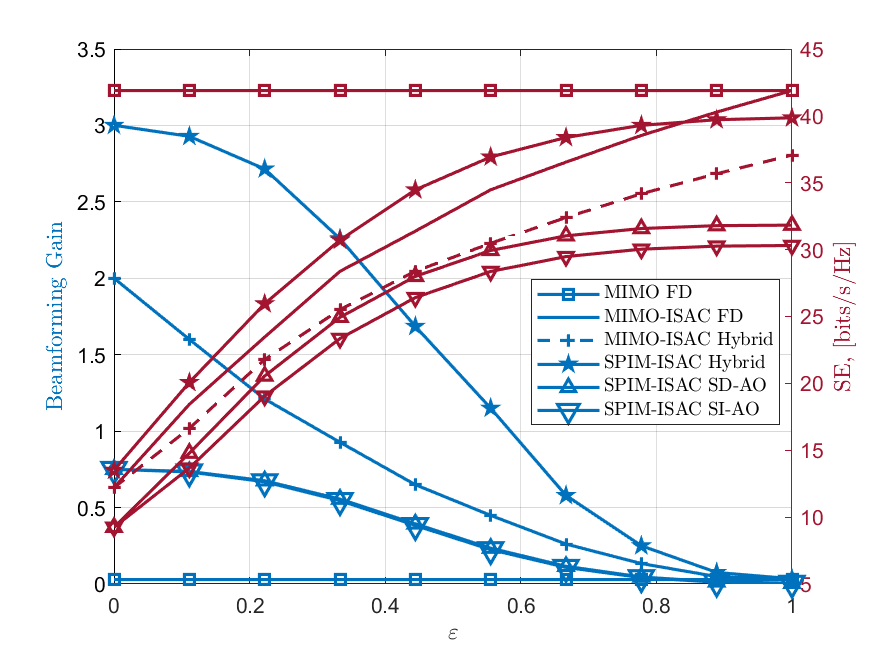} }\\
			\subfloat[]{\includegraphics[draft=false,width=\columnwidth]{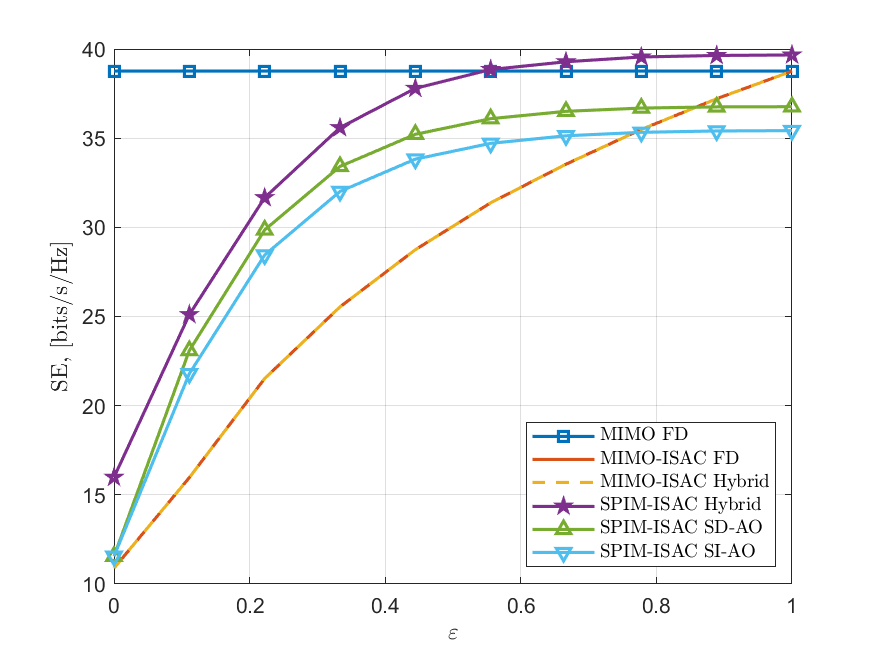} }
			\caption{(a) SE and beamforming gain performance versus $\varepsilon$ when $(L,L_\mathrm{S}) = (8,3)$. (b) SE versus $\varepsilon$  for  $(L,L_\mathrm{S}) = (12,10)$, when $\mathrm{SNR} = 0$ dB.
			}
			
			\label{fig_SE_eta}
		\end{figure}
		
		

		Fig.~\ref{fig_SE_SNR} shows the SE with respect to SNR  when the radar-communications trade-off parameter is $\varepsilon =0.5$. {The MIMO-ISAC and SPIM-ISAC beamformers are computed via (\ref{MI_mmwave_only}) and (\ref{MI_SPIM}), respectively.} The MIMO FD beamformer constitutes a benchmark while the JRC beamformer provides a trade-off between radar and communications. {Specifically, the computation of MIMO FD and MIMO-ISAC FD beamformers are obtained via $\mathbf{F}_\mathrm{opt}[m]$ and $\mathbf{F}_\mathrm{CR}[m]$ in (\ref{Fcr}), respectively.}  We observe from Fig.~\ref{fig_SE_SNR} that a significant improvement is achieved in SE with our proposed SPIM-ISAC hybrid beamforming approach compared to MIMO-ISAC even with JRC beamformer. Although hybrid beamformers are employed, the proposed SPIM approach provides higher SE thanks to additional transmitted information bits via SPIM. {Note that similar observations have also been made in the literature~\cite{spim_GBMM_Guo2019Jul,spim_bounds_JSTSP_Wang2019May} which, however, involves communications-only MIMO system design.} When we compare the proposed SD- and SI-AO beamforming techniques, the former exhibits higher SE as compared to the latter as the former takes advantages of SD implementation. Thus, the SD-AO beamformer is resilient to beam-split with the cost of employing $(M-1)N_\mathrm{T}N_\mathrm{RF}$ phase shifters. {In contrast, the SD hybrid beamformers (i.e., MIMO SD hybrid and SPIM-ISAC hybrid) yield higher SE than the SD AO beamformers because they employ $MN_\mathrm{T}N_\mathrm{RF}$ phase shifters. Nevertheless, the proposed SPIM-ISAC hybrid beamformer attains very close performance to SPIM-ISAC SD hybrid beamformer because of its beam-split compensation.   }

		\begin{figure}[t]
			\centering
			\subfloat[]{\includegraphics[draft=false,width=\columnwidth]{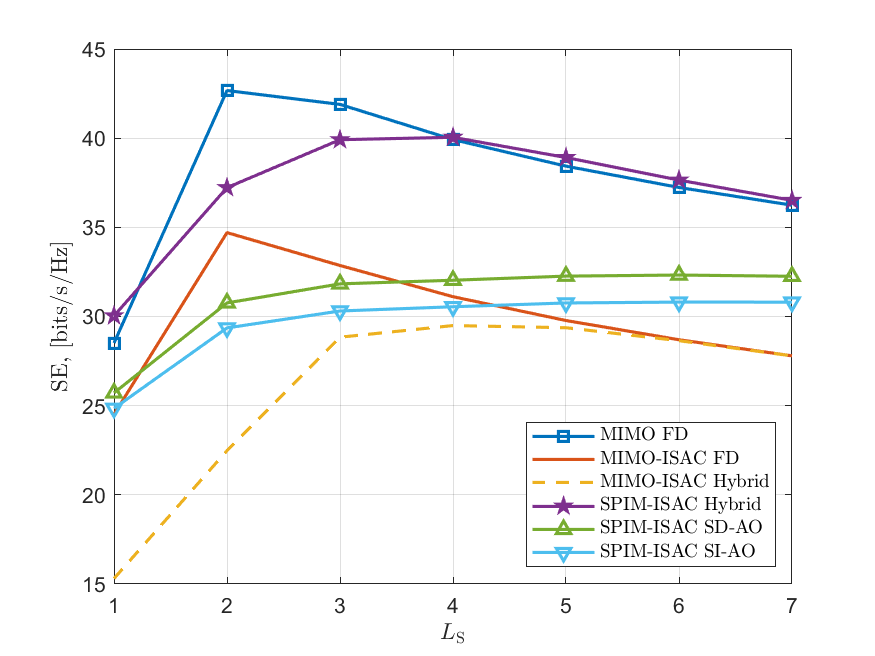} } \\
			\subfloat[]{\includegraphics[draft=false,width=\columnwidth]{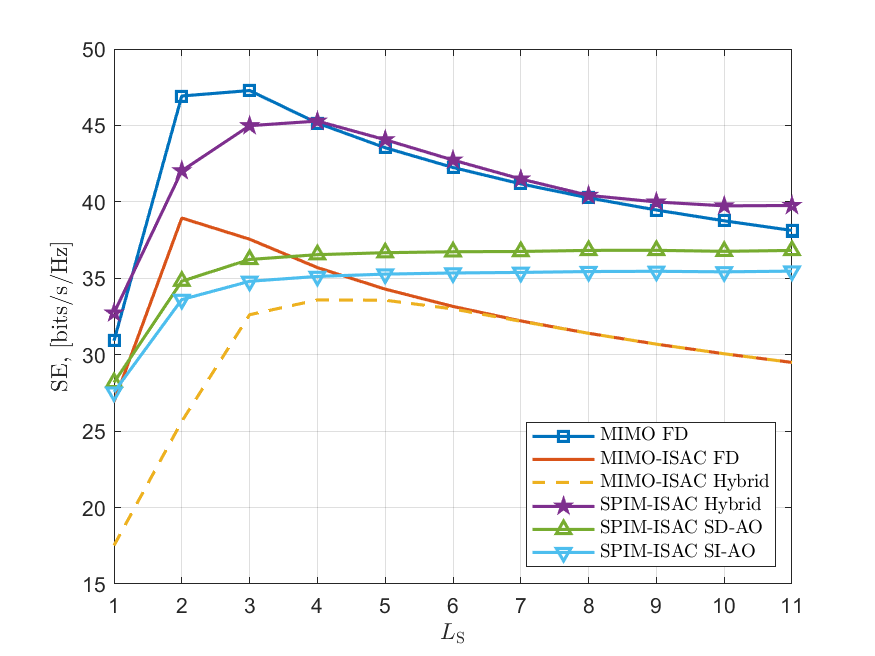} }
			\caption{SE versus $L_\mathrm{S}$ for (a) $L=8$ and (b) $L=12$, respectively, when $\varepsilon = 1$ and $\mathrm{SNR} = 0$ dB.
			}
			
			\label{fig_SE_Ls}
		\end{figure}

		
		Fig.~\ref{fig_SE_eta} shows the system  performance with respect to the trade-off parameter $\varepsilon$ for $(L,L_\mathrm{S}) = (8,3)$. Specifically, Fig.~\ref{fig_SE_eta}(a) explicitly demonstrates the trade-off on both communications (SE) and radar (beamforming gain evaluated at target directions via (\ref{beamPattern})) metrics. We can see that as $\varepsilon \rightarrow 1$, SE of the competing algorithms increases whereas the beamforming gain decreases.  Furthermore, the performance of the proposed SPIM approaches improves as $\varepsilon \rightarrow 1$, as expected, and they demonstrate even higher SE than the MIMO-ISAC JRC design, e.g, when approximately $\varepsilon >0.8$. In Fig.~\ref{fig_SE_eta}(b)  the SE performance is presented for $(L,L_\mathrm{S}) = (12,10)$. Compared to the case $(L,L_\mathrm{S}) = (8,3)$ in Fig.~\ref{fig_SE_eta}(a), the results in Fig.~\ref{fig_SE_eta}(b) yields higher SE for all of the methods. Notably,  the proposed SPIM-ISAC hybrid beamformer achieves much higher SE than that of MIMO FD beamformer for $\varepsilon \geq 0.6$ thanks to additional SE provided via SPIM with higher $L$ and $L_\mathrm{S}$. 
		Fig.~\ref{fig_SE_eta}(b) also shows that SD- and SI-AO beamformers exhibit higher SE performance than that of MIMO-ISAC designs (hybrid and JRC) for up to $\varepsilon \geq 0.8$, however its performance falls behind the MIMO-ISAC  beamformers as $\varepsilon$ further increases. This is because the AO beamformers have limited performance due to the absence of baseband beamformers.

		\begin{figure}[t]
			\centering
			{\includegraphics[draft=false,width=.99\columnwidth]{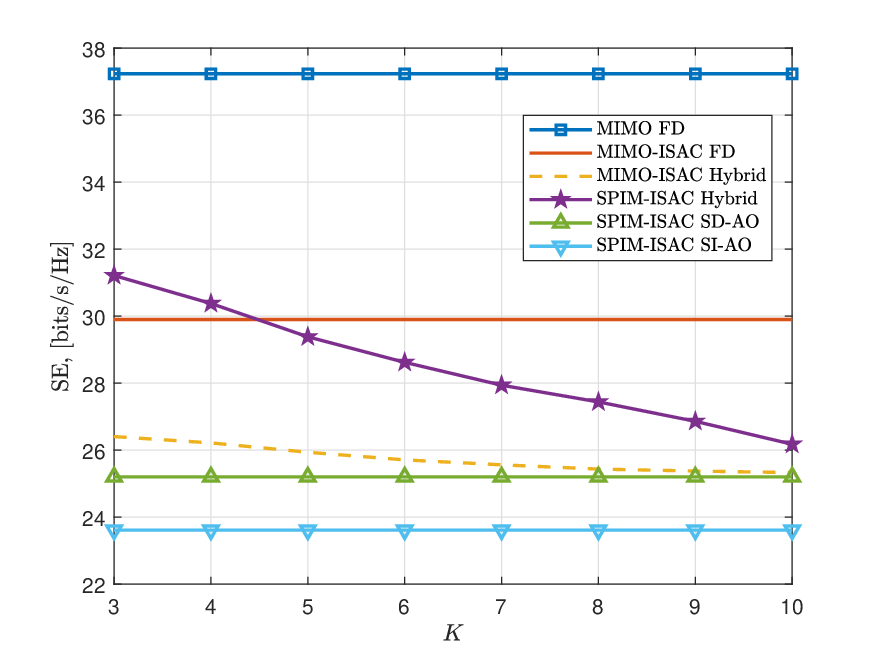} } 
			\caption{SE versus $K$ when $P = L + \bar{K}$ and $N_\mathrm{RF} = L_\mathrm{S} + \bar{K}$, where $\bar{K} = 3$, and $\mathrm{SNR} = 0$ dB.
			}
			
			\label{fig_SE_K}
		\end{figure}

		\begin{figure}[t]
			\centering
			{\includegraphics[draft=false,width=.99\columnwidth]{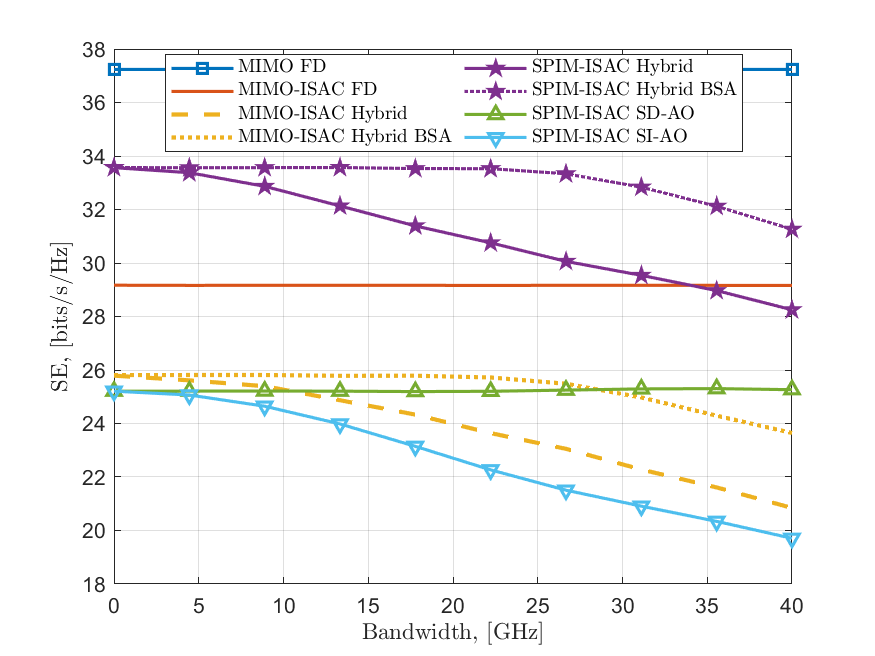} } 
			\caption{SE versus bandwidth for THz system with $f_c = 300$ GHz when $\mathrm{SNR} = 0$ dB and  $\varepsilon = 0.5$.
			}
			
			\label{fig_SE_BW}
		\end{figure}

		\begin{figure*}[t]
			\centering
			\subfloat[]{\includegraphics[draft=false,width=.33\textwidth]{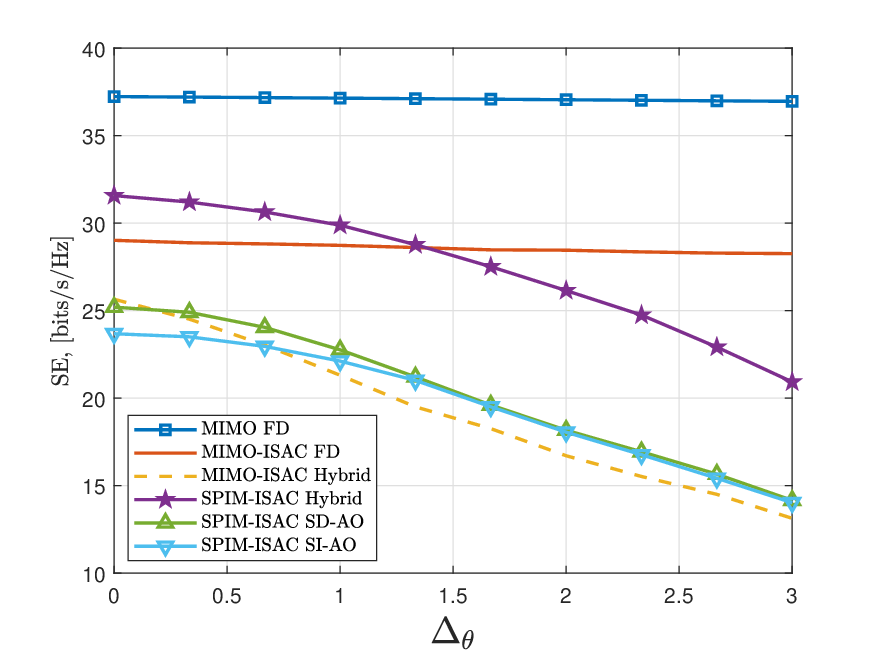} }
			\subfloat[]{\includegraphics[draft=false,width=.33\textwidth]{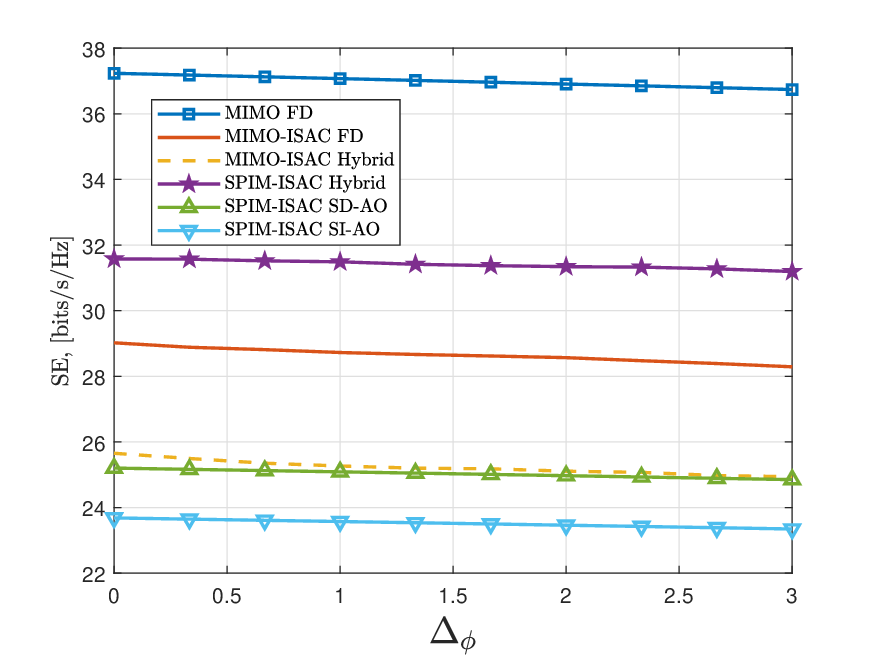} }
			\subfloat[]{\includegraphics[draft=false,width=.33\textwidth]{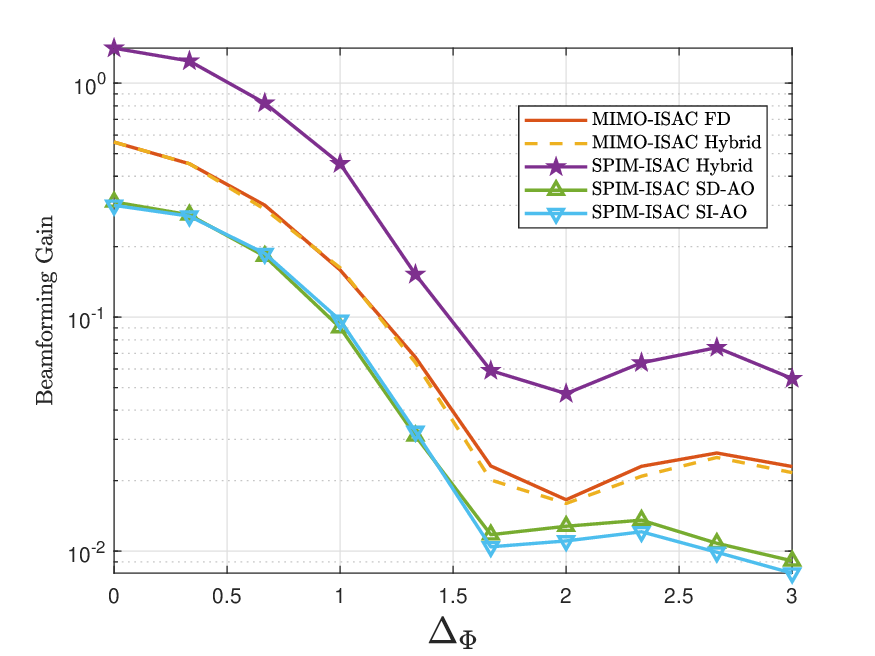} }
			\caption{SE versus the mismatch on (a) DoD $\Delta_{\theta}$ and (b) DoA $\Delta_{\phi}$, and beamforming gain versus mismatch on (c) target DoAs $\Delta_{\Phi}$  when $\varepsilon = 0.5$ and $\mathrm{SNR} = 0$ dB.
			}
			
			\label{fig_SE_uncertainty}
		\end{figure*}

		Fig.~\ref{fig_SE_Ls} shows the SE performance with respect to number of selected SPIM paths $L_\mathrm{S}$ for (a) $L=8$ and (b) $L=12$, respectively, when $\varepsilon = 0.5$ and $\mathrm{SNR} = 0$ dB. As $L_\mathrm{S}$ increases, the performance of AO beamformers is saturated while the SE of the hybrid beamformers first increases then decreases. This is the result of sparse mmWave channel and the unoptimized power allocation of the baseband beamformers to less/more important path components, which can be compensated via \textit{multi-mode beamforming} techniques~\cite{alkhateeb2016frequencySelective}. When compared to the cases $L=8$ and $L=12$, higher SE is achieved for all algorithms in the latter. Furthermore, the proposed SPIM-ISAC hybrid beamformer achieves higher SE than that of MIMO FD beamformer for $L_\mathrm{S}\geq 4$ ($L_\mathrm{S}\geq 4$) when $L=8$ ($L=12$).   Note that the performance improvement obtained from SPIM-ISAC is limited to the number of available spatial paths in the environment. Thus, one cannot always achieve higher SE by employing  SM over more paths since the achieved SE is also limited by the number of RF chains at the cost of higher hardware complexity.

		{We also present the SE performance with respect to the number of targets $K$ in Fig.~\ref{fig_SE_K} while the number of nodes in the switching network at the BS (see Fig.~\ref{fig_BS}) is kept fixed to $P = L + \bar{K}$, where $\bar{K}=3$. We see that the performance of the SPIM-ISAC hybrid beamformer degrades as $K$ increases since the beamformer becomes unable to serve all $K$ targets when $K> \bar{K}$.    }

		In order to demonstrate the performance of the proposed BSA hybrid beamforming technique, the SE of the beamformers are given in Fig.~\ref{fig_SE_BW} with respect to the bandwidth $B \in [0, 40]$ GHz. {In this experiment, we consider the THz scenario with $f_c=300$ GHz, $L=5$ and $L_\mathrm{S}=3$ while the remaining simulation parameters are kept fixed.} However, similar results can also be achieved if the same signal model is used for the mmWave scenario with $d = \frac{\lambda_c}{2}$, which corresponds to $f_c = \frac{300}{5} = 60$ GHz and $B\in [0,8]$ GHz. We can see from Fig.~\ref{fig_SE_BW} that the FD beamformers are not affected by the beam-split since they do not include analog components. The SD-AO beamformer also provides a robust performance against beam-split at the cost reduced SE since it is implemented in SD manner without baseband components. The proposed BSA approach is employed in MIMO-ISAC and SPIM-ISAC hybrid beamformers. We can see that the performance of proposed BSA approach yields robust performance up to approximately $B \leq 30$ GHz. Note that the performance of the proposed BSA hybrid beamforming approach is limited by the number of RF chains. In particular, the beam-split can be fully mitigated only if  $\mathbf{F}_\mathrm{RF}^{(i)} {\mathbf{F}_\mathrm{RF}^{(i)}}^\dagger  = \mathbf{I}_{\mathrm{N}_\mathrm{T}}$, which requires $N_\mathrm{RF} = N_\mathrm{T}$. Nevertheless, the proposed approach has satisfactory performance without employing additional hardware components, e.g., TD networks. In addition, the performance loss because of beam-split is further compensated thanks to additional SE gain via SPIM.

		Fig~\ref{fig_SE_uncertainty} presents the performance analysis with respect to the angular mismatch in the estimated DoD and DoA angles (i.e., $\theta$ and $\phi$) of the communications user as well as the DoA angles of the radar targets (i.e., $\Phi$). During simulations, the mismatch DoD/DoA angles are generated as  $\breve{\kappa} \sim \mathcal{N}(\kappa + \Delta_{\kappa}, (0.1\Delta_{\kappa})^2)$, where $\kappa\in \{\theta, \phi, \Phi \} $. Fig.~\ref{fig_SE_uncertainty}(a) and Fig.\ref{fig_SE_uncertainty}(b) show the SE with respect to $\Delta_{\theta}$ and $\Delta_{\phi}$, respectively, while Fig.~\ref{fig_SE_uncertainty}(c) shows the beamforming gain with respect to $\Delta_{\Phi}$. We can see that the SE is more tolerable to the mismatch in $\phi$ than that of $\theta$ because of $N_\mathrm{R}< N_\mathrm{T}$.

		\begin{figure}[t]
			\centering
			{\includegraphics[draft=false,width=.99\columnwidth]{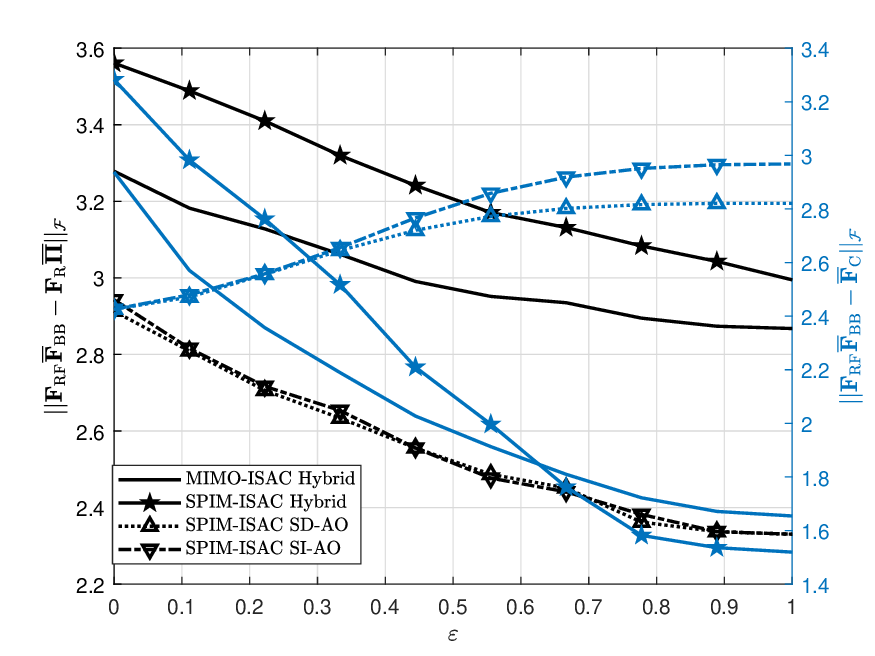} } 
			\caption{Beamforming error for JRC with respect to $\varepsilon$ when  $(L,L_\mathrm{S}) = (8,3)$ and  $\mathrm{SNR} = 0$ dB.
			}
			
			\label{fig_JRC_eta}
		\end{figure}
		
		In order to present the system performance with respect to a joint metric for the beamformers, Fig.~\ref{fig_JRC_eta} shows the JRC performance in terms of the error between the hybrid beamformers and the communication/radar-only beamformers, i.e., $|| \mathbf{F}_\mathrm{RF} \overline{\mathbf{F}}_\mathrm{BB}  -\overline{\mathbf{F}}_\mathrm{C} ||_\mathcal{F}$ and $|| \mathbf{F}_\mathrm{RF}\overline{ \mathbf{F}}_\mathrm{BB} - \mathbf{F}_\mathrm{R}\overline{\mathbf{\Pi}}  ||_\mathcal{F}$, respectively. we can see that the proposed SPIM-ISAC provides less error with respect to $\overline{ \mathbf{F}}_\mathrm{C}$ and $\mathbf{F}_\mathrm{R}\overline{\boldsymbol{\Pi}}$ as compared to the competing beamformers.

		Finally, we present the beampattern of the proposed SPIM-ISAC hybrid beamformer in Fig.~\ref{fig_BP} for $\eta = \{0,0.3,0.5,0.8,1\}$ when only $S = 2$ ($i \in \{1,2\}$) spatial patterns are used. In this scenario, $K=1$ and $(L,L_\mathrm{S}) = (2,1)$. The target is located at  $\Phi_1 = 40^\circ$ while the BS receives the incoming paths from the communications user at $\theta_1 = 50^\circ$ ($i =2$) and $\theta_2 = 60^\circ$ ($i =2$), respectively. The beampattern becomes suppressed at the target direction when $\eta \rightarrow 1$. Conversely, the beampattern at the user locations is minimized when $\eta \rightarrow 0$. This illustrates the effectiveness of our proposed SPIM-ISAC approach. 
		
		\begin{figure}[t]
			\centering
			{\includegraphics[draft=false,width=.99\columnwidth]{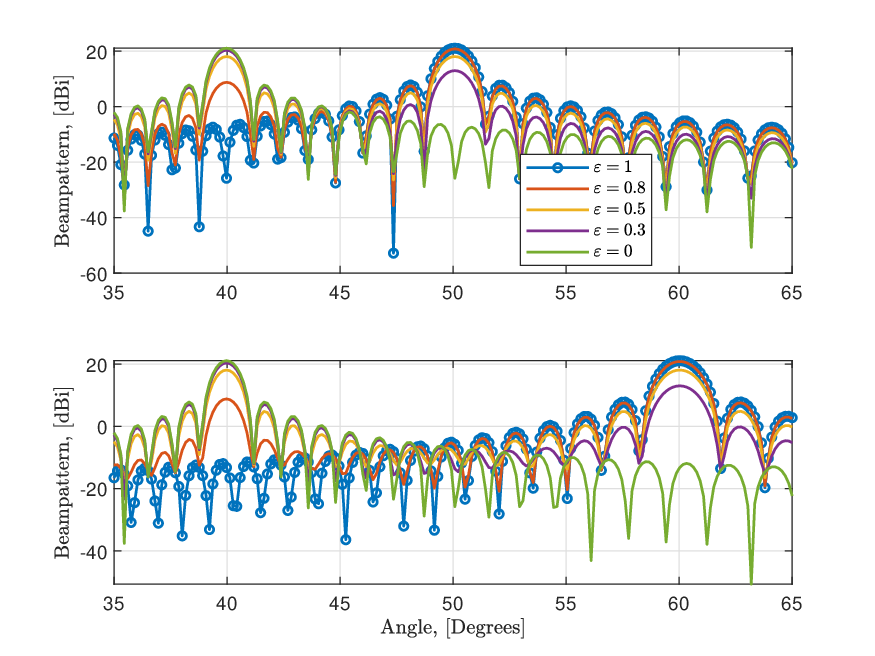} } 
			\caption{Beampattern for $K=1$, $(L,L_\mathrm{S}) = (2,1)$ when $(\Phi_1, \theta_1)$ is (top) $(40^\circ, 50^\circ)$ ($i=1$) and (bottom) $(40^\circ,60^\circ)$ $(i=2)$, respectively, for various values of $\varepsilon = \{0,0.3,0.5,0.8,1\}$. 
			}
			
			\label{fig_BP}
		\end{figure}
		\section{Summary}
		\label{sec:Conc}
		We introduced a SPIM framework for ISAC, wherein the hybrid beamformers are designed by exploiting the spatial scattering paths between the BS and the communications user. We showed that a significant performance improvement is achieved via SPIM-ISAC compared to conventional MIMO-ISAC, wherein only the strongest path is selected for beamformer design. We introduced a family of beamforming techniques: hybrid, BSA hybrid, SI-AO, and SD-AO. We analyzed their respective trade-offs in terms of SE, beamforming gain, and hardware complexity. 
		The proposed SPIM-ISAC hybrid beamformer takes advantage of employing baseband beamformer and its BSA hybrid beamforming technique achieves higher SE than the AO beamformers. 
		Furthermore, the proposed SPIM-ISAC hybrid beamforming approach exhibits significant spectral efficiency performance even higher than that of the usage of MIMO-ISAC FD beamformers in the presence of beam-split. The proposed SPIM-ISAC approach is a viable solution to the performance loss resulting from beam-split for both mmWave and THz systems.

		\appendices

		\section{Proof of Lemma~\ref{lemma1}}
		\label{appen2ArrayGain}
		
		The array gain varies across the whole bandwidth as
		\begin{align}
			\label{arrayGain2}
			A_G({\Phi},{m}) = \frac{|\mathbf{a}_\mathrm{T}^\textsf{H}(\Phi)  \mathbf{a}_\mathrm{T}(\Phi_m)|^2}{N_\mathrm{T}^2}.
		\end{align}
		By using (\ref{steringVec_aT}), (\ref{arrayGain2}) is rewritten  as
		\begin{align}
			&A_G({\Phi},{m})= \frac{1}{N_\mathrm{T}^2} \left| \sum_{n_1 =1}^{N_\mathrm{T}}  \sum_{n_2=1}^{N_\mathrm{T}} e^{-\mathrm{j} \pi  \left( (n_1-1){\Phi_m} - (n_2-1)\frac{\lambda_c\Phi}{\lambda_m}\right)    }  \hspace{-3pt}  \right| ^2 \hspace{-3pt} \nonumber \\
			& \hspace{-3pt}= \hspace{-3pt}\frac{1}{N_\mathrm{T}^2}  \hspace{-3pt} \left| \hspace{-3pt}\sum_{n = 0}^{N_\mathrm{T}-1} e^{-\mathrm{j}2\pi n d \left( \frac{\Phi_m}{\lambda_c} - \frac{\Phi}{\lambda_m}  \right)     }  \hspace{-3pt}  \right|^2  \hspace{-3pt}
			\hspace{-3pt}	=  \hspace{-3pt} \frac{1}{N_\mathrm{T}^2} \hspace{-3pt} \left| \hspace{-3pt}\sum_{n = 0}^{N_\mathrm{T}-1} e^{-\mathrm{j}2\pi n d \frac{(f_c \Phi_m - f_m\Phi) }{c_0}     }  \hspace{-3pt}  \right|^2 \hspace{-3pt} \nonumber\\
			& \hspace{-3pt}=  \hspace{-3pt} \hspace{-3pt}\frac{1}{N_\mathrm{T}^2}  \hspace{-3pt}\left| \hspace{-3pt} \frac{1 - e^{-\mathrm{j}2\pi N_\mathrm{T}d \frac{(f_c \Phi_m - f_m\Phi)}{c_0}    }}{1 - e^{-\mathrm{j}2\pi d\frac{(f_c \Phi_m - f_m\Phi)}{c_0}  }}   \hspace{-3pt} \right|^2  \hspace{-3pt}
			\hspace{-3pt}	=  \hspace{-3pt}\frac{1}{N_\mathrm{T}^2}\left| \frac{\sin (\pi N_\mathrm{T}\mu_m )}{\sin (\pi \gamma_m )}     \hspace{-3pt}\right|^2  \hspace{-3pt} \hspace{-3pt}=  \hspace{-3pt}|\xi( \mu_m )|^2 \hspace{-3pt}, \label{arrayGain}
		\end{align}
		where    $\mu_m = d\frac{(f_c \Phi_m - f_m\Phi)}{c_0}   $. The array gain in (\ref{arrayGain}) implies that most of the power is focused only on a small portion of the beamspace due to the power-focusing capability of $\xi(a)$, which substantially reduces across the subcarriers as $|f_m - f_c|$ increases. Furthermore, $|\xi( \mu_m )|^2$ gives peak when $\mu_m = 0$, i.e.,  $f_c \Phi_m - f_m\Phi= 0$. Thus, we have $ \Phi_m = \eta_m \Phi$, which completes the proof.  \qed

		
		\footnotesize
		\bibliographystyle{IEEEtran}
		\bibliography{IEEEabrv,references_116}

\begin{thebibliography}{10}
\providecommand{\url}[1]{#1}
\csname url@samestyle\endcsname
\providecommand{\newblock}{\relax}
\providecommand{\bibinfo}[2]{#2}
\providecommand{\BIBentrySTDinterwordspacing}{\spaceskip=0pt\relax}
\providecommand{\BIBentryALTinterwordstretchfactor}{4}
\providecommand{\BIBentryALTinterwordspacing}{\spaceskip=\fontdimen2\font plus
\BIBentryALTinterwordstretchfactor\fontdimen3\font minus
  \fontdimen4\font\relax}
\providecommand{\BIBforeignlanguage}[2]{{%
\expandafter\ifx\csname l@#1\endcsname\relax
\typeout{** WARNING: IEEEtran.bst: No hyphenation pattern has been}%
\typeout{** loaded for the language `#1'. Using the pattern for}%
\typeout{** the default language instead.}%
\else
\language=\csname l@#1\endcsname
\fi
#2}}
\providecommand{\BIBdecl}{\relax}
\BIBdecl

\bibitem{elbir_SPIM_MMWAVE_RadarConf_Elbir2022Nov}
A.~M. Elbir, K.~V. Mishra, A.~Celik, and A.~M. Eltawil, ``{Millimeter-Wave
  Radar Beamforming with Spatial Path Index Modulation Communications},'' in
  \emph{{2023 IEEE Radar Conference (RadarConf23)}}.\hskip 1em plus 0.5em minus
  0.4em\relax IEEE, May 2023, pp. 1--6.

\bibitem{mishra2019toward}
K.~V. Mishra, M.~R.~B. Shankar, V.~Koivunen, B.~Ottersten, and S.~A. Vorobyov,
  ``Toward millimeter-wave joint radar communications: {A} signal processing
  perspective,'' \emph{IEEE Signal Process. Mag.}, vol.~36, no.~5, pp.
  100--114, 2019.

\bibitem{elbir_thz_jrc_Magazine_Elbir2022Aug}
A.~M. Elbir, K.~V. Mishra, S.~Chatzinotas, and M.~Bennis, ``Terahertz-band
  integrated sensing and communications: {C}hallenges and opportunities,''
  \emph{arXiv preprint arXiv:2208.01235}, 2022.

\bibitem{mishra2023signal}
K.~V. Mishra, I.~Bilik, J.~Tabrikian, and A.~P. Petropulu, ``Signal processing
  for terahertz-band automotive radars: {E}xploring the next frontier,''
  \emph{arXiv preprint}, 2023.

\bibitem{jrc_TCOM_Liu2020Feb}
F.~Liu, C.~Masouros, A.~P. Petropulu, H.~Griffiths, and L.~Hanzo, ``Joint radar
  and communication design: {A}pplications, state-of-the-art, and the road
  ahead,'' \emph{IEEE Trans. Commun.}, vol.~68, no.~6, pp. 3834--3862, 2020.

\bibitem{elbir2021JointRadarComm}
A.~M. Elbir, K.~V. Mishra, and S.~Chatzinotas, ``Terahertz-band joint
  ultra-massive {MIMO} radar-communications: {M}odel-based and model-free
  hybrid beamforming,'' \emph{IEEE J. Sel. Top. Signal Process.}, vol.~15,
  no.~6, pp. 1468--1483, 2021.

\bibitem{liu2020co}
J.~Liu, K.~V. Mishra, and M.~Saquib, ``Co-designing statistical {MIMO} radar
  and in-band full-duplex multi-user {MIMO} communications,'' \emph{arXiv
  preprint arXiv:2006.14774}, 2020.

\bibitem{duggal2020doppler}
G.~Duggal, S.~Vishwakarma, K.~V. Mishra, and S.~S. Ram, ``Doppler-resilient
  802.11ad-based ultrashort range automotive joint radar-communications
  system,'' \emph{IEEE Trans. Aerosp. Electron. Syst.}, vol.~56, no.~5, pp.
  4035--4048, 2020.

\bibitem{sedighi2021localization}
S.~Sedighi, K.~V. Mishra, M.~R.~B. Shankar, and B.~Ottersten, ``Localization
  with one-bit passive radars in narrowband internet-of-things using
  multivariate polynomial optimization,'' \emph{IEEE Trans. Signal Process.},
  vol.~69, pp. 2525--2540, 2021.

\bibitem{thz_isac5_Petrov2019May}
V.~Petrov, G.~Fodor, J.~Kokkoniemi, D.~Moltchanov, J.~Lehtomaki, S.~Andreev,
  Y.~Koucheryavy, M.~Juntti, and M.~Valkama, ``{On Unified Vehicular
  Communications and Radar Sensing in Millimeter-Wave and Low Terahertz
  Bands},'' \emph{IEEE Wireless Commun.}, vol.~26, no.~3, pp. 146--153, 2019.

\bibitem{beamSquint_FeiFei_Wang2019Oct}
B.~Wang, M.~Jian, F.~Gao, G.~Y. Li, and H.~Lin, ``Beam squint and channel
  estimation for wideband {mmWave} massive {MIMO-OFDM} systems,'' \emph{IEEE
  Trans. Signal Process.}, vol.~67, no.~23, pp. 5893--5908, 2019.

\bibitem{delayPhasePrecoding_THz_Dai2022Mar}
L.~Dai, J.~Tan, Z.~Chen, and H.~V. Poor, ``Delay-phase precoding for wideband
  {THz} massive {MIMO},'' \emph{IEEE Trans. Wireless Commun.}, vol.~21, no.~9,
  pp. 7271--7286, 2022.

\bibitem{trueTimeDelayBeamSquint}
F.~Gao, B.~Wang, C.~Xing, J.~An, and G.~Y. Li, ``Wideband beamforming for
  hybrid massive {MIMO} terahertz communications,'' \emph{IEEE J. Sel. Areas
  Commun.}, vol.~39, no.~6, pp. 1725--1740, 2021.

\bibitem{elbir_THZ_CE_ArrayPerturbation_Elbir2022Aug}
A.~M. Elbir, W.~Shi, A.~K. Papazafeiropoulos, P.~Kourtessis, and
  S.~Chatzinotas, ``{Terahertz-Band Channel and Beam Split Estimation via Array
  Perturbation Model},'' \emph{IEEE Open J. Commun. Soc.}, vol.~4, pp.
  892--907, Mar. 2023.

\bibitem{heath2016overview}
R.~W. Heath, N.~Gonz{\ifmmode\acute{a}\else\'{a}\fi}lez-Prelcic, S.~Rangan,
  W.~Roh, and A.~M. Sayeed, ``An overview of signal processing techniques for
  millimeter wave {MIMO} systems,'' \emph{IEEE J. Sel. Top. Signal Process.},
  vol.~10, no.~3, pp. 436--453, 2016.

\bibitem{elbir2022Nov_Beamforming_SPM}
A.~M. Elbir, K.~V. Mishra, S.~A. Vorobyov, and R.~W. Heath, ``{Twenty-Five
  Years of Advances in Beamforming: From convex and nonconvex optimization to
  learning techniques},'' \emph{IEEE Signal Process. Mag.}, vol.~40, no.~4, pp.
  118--131, Jun. 2023.

\bibitem{alkhateeb2016frequencySelective}
A.~{Alkhateeb} and R.~W. {Heath}, ``Frequency selective hybrid precoding for
  limited feedback millimeter wave systems,'' \emph{IEEE Trans. Commun.},
  vol.~64, no.~5, pp. 1801--1818, 2016.

\bibitem{beamSquintAwareHB_SD_You2022Aug}
L.~You, X.~Qiang, C.~G. Tsinos, F.~Liu, W.~Wang, X.~Gao, and B.~Ottersten,
  ``Beam squint-aware integrated sensing and communications for hybrid massive
  {MIMO} {LEO} satellite systems,'' \emph{IEEE J. Sel. Areas Commun.}, vol.~40,
  no.~10, pp. 2994--3009, 2022.

\bibitem{indexMod_Survey_Mao2018Jul}
T.~Mao, Q.~Wang, Z.~Wang, and S.~Chen, ``Novel index modulation techniques: {A}
  survey,'' \emph{IEEE Commun. Surv. Tutorials}, vol.~21, no.~1, pp. 315--348,
  2018.

\bibitem{hodge2023index}
J.~A. Hodge, K.~V. Mishra, B.~M. Sadler, and A.~I. Zaghloul, ``Reconfigurable
  intelligent surfaces for {6G} wireless networks using index-modulated
  metasurface transceivers,'' \emph{IEEE J. Sel. Topics Signal Process.}, 2023,
  in press.

\bibitem{hodge2020intelligent}
J.~A. Hodge, K.~V. Mishra, and A.~I. Zaghloul, ``Intelligent time-varying
  metasurface transceiver for index modulation in {6G} wireless networks,''
  \emph{IEEE Antennas Wirel. Propag. Lett.}, vol.~19, no.~11, pp. 1891--1895,
  2020.

\bibitem{antenna_grouping_SM}
L.~{He}, J.~{Wang}, and J.~{Song}, ``Spatial modulation for more spatial
  multiplexing: {RF}-chain-limited generalized spatial modulation aided
  {mm-wave} {MIMO} with hybrid precoding,'' \emph{{IEEE} Trans. Commun.},
  vol.~66, no.~3, pp. 986--998, 2018.

\bibitem{hodge2019reconfigurable}
J.~A. Hodge, K.~V. Mishra, and A.~I. Zaghloul, ``Reconfigurable metasurfaces
  for index modulation in {5G} wireless communications,'' in \emph{IEEE Int.
  Appl. Comput. Electromagn. Soc. Symp.}, 2019, pp. 1--2.

\bibitem{jrc_spim_sm_Ma2021Feb}
D.~Ma, N.~Shlezinger, T.~Huang, Y.~Shavit, M.~Namer, Y.~Liu, and Y.~C. Eldar,
  ``Spatial modulation for joint radar-communications systems: {D}esign,
  analysis, and hardware prototype,'' \emph{IEEE Trans. Veh. Technol.},
  vol.~70, no.~3, pp. 2283--2298, 2021.

\bibitem{spim_bounds_JSTSP_Wang2019May}
J.~Wang, L.~He, and J.~Song, ``Towards higher spectral efficiency: {S}patial
  path index modulation improves millimeter-wave hybrid beamforming,''
  \emph{IEEE J. Sel. Top. Signal Process.}, vol.~13, no.~6, pp. 1348--1359,
  2019.

\bibitem{spim_BIM_TVT_Ding2018Mar}
Y.~Ding, V.~Fusco, A.~Shitvov, Y.~Xiao, and H.~Li, ``Beam index modulation
  wireless communication with analog beamforming,'' \emph{IEEE Trans. Veh.
  Technol.}, vol.~67, no.~7, pp. 6340--6354, 2018.

\bibitem{spim_GBM_Gao2019Jul}
S.~Gao, X.~Cheng, and L.~Yang, ``Spatial multiplexing with limited {RF} chains:
  {G}eneralized beamspace modulation ({GBM}) for {mmWave} massive {MIMO},''
  \emph{IEEE J. Sel. Areas Commun.}, vol.~37, no.~9, pp. 2029--2039, 2019.

\bibitem{spim_onGSM_He2017Sep}
L.~He, J.~Wang, and J.~Song, ``On generalized spatial modulation aided
  millimeter wave {MIMO}: {S}pectral efficiency analysis and hybrid precoder
  design,'' \emph{IEEE Trans. Wireless Commun.}, vol.~16, no.~11, pp.
  7658--7671, 2017.

\bibitem{spim_lowComplexGSM_Shi2021Jan}
X.~Shi, J.~Wang, C.~Pan, and J.~Song, ``Low-complexity hybrid precoding
  algorithm based on log-det expansion for {GenSM}-aided {MmWave} {MIMO}
  system,'' \emph{IEEE Trans. Veh. Technol.}, vol.~70, no.~2, pp. 1554--1564,
  2021.

\bibitem{spim_GBMM_Guo2019Jul}
S.~Guo, H.~Zhang, P.~Zhang, P.~Zhao, L.~Wang, and M.-S. Alouini, ``Generalized
  beamspace modulation using multiplexing: {A} breakthrough in {mmWave}
  {MIMO},'' \emph{IEEE J. Sel. Areas Commun.}, vol.~37, no.~9, pp. 2014--2028,
  2019.

\bibitem{waveforDesign_SpectralMod1_Yu2022Apr}
X.~Yu, X.~Yao, J.~Yang, L.~Zhang, L.~Kong, and G.~Cui, ``{Integrated Waveform
  Design for MIMO Radar and Communication via Spatio-Spectral Modulation},''
  \emph{IEEE Trans. Signal Process.}, vol.~70, pp. 2293--2305, Apr. 2022.

\bibitem{waveforDesign_SpectralMod2_Yang2022Dec}
J.~Yang, Y.~Tan, X.~Yu, G.~Cui, and D.~Zhang, ``{Waveform Design for Watermark
  Framework Based DFRC System With Application on Joint SAR Imaging and
  Communication},'' \emph{IEEE Trans. Geosci. Remote Sens.}, vol.~61, pp.
  1--14, Dec. 2022.

\bibitem{Mao2021ISAC}
T.~Mao and Z.~Wang, ``Terahertz wireless communications with flexible index
  modulation aided pilot design,'' \emph{IEEE Journal on Selected Areas in
  Communications}, vol.~39, no.~6, pp. 1651--1662, Jun. 2021.

\bibitem{spim_secureSM_SubarraySelection_Shu2020Nov}
F.~Shu, X.~Jiang, X.~Liu, L.~Xu, G.~Xia, and J.~Wang, ``Precoding and transmit
  antenna subarray selection for secure hybrid spatial modulation,'' \emph{IEEE
  Trans. Wireless Commun.}, vol.~20, no.~3, pp. 1903--1917, 2020.

\bibitem{spim_secureGSM_Yang2020Oct}
P.~Yang and X.~Qiu, ``Hybrid precoding aided secure generalized spatial
  modulation in millimeter wave {MIMO} systems,'' \emph{IEEE Commun. Lett.},
  vol.~25, no.~2, pp. 397--401, 2020.

\bibitem{spim_secureGSM_FiniteAlphabet_Xia2020Dec}
G.~Xia, Y.~Lin, X.~Zhou, W.~Zhang, F.~Shu, and J.~Wang, ``Hybrid precoding
  design for secure generalized spatial modulation with finite-alphabet
  inputs,'' \emph{IEEE Trans. Commun.}, vol.~69, no.~4, pp. 2570--2584, 2020.

\bibitem{spim_EnergyE_Raafat2019Dec}
A.~Raafat, M.~Sefun{\ifmmode\mbox{\c{c}}\else\c{c}\fi}, A.~Agustin, J.~Vidal,
  E.~A. Jorswieck, and Y.~Corre, ``Energy efficient transmit-receive hybrid
  spatial modulation for large-scale {MIMO} systems,'' \emph{IEEE Trans.
  Commun.}, vol.~68, no.~3, pp. 1448--1463, 2019.

\bibitem{spim_GSM_CE_Chu2019May}
H.~Chu, L.~Zheng, and X.~Wang, ``Super-resolution {mmWave} channel estimation
  for generalized spatial modulation systems,'' \emph{IEEE J. Sel. Top. Signal
  Process.}, vol.~13, no.~6, pp. 1336--1347, 2019.

\bibitem{spim_IRS_SM_Yurduseven2020Aug}
O.~Yurduseven, S.~D. Assimonis, and M.~Matthaiou, ``Intelligent reflecting
  surfaces with spatial modulation: {A}n electromagnetic perspective,''
  \emph{IEEE Open J. Commun. Soc.}, vol.~1, pp. 1256--1266, 2020.

\bibitem{spim_IRS_SM_antennaSelection_Ma2020Jul}
T.~Ma, Y.~Xiao, X.~Lei, P.~Yang, X.~Lei, and O.~A. Dobre, ``Large intelligent
  surface assisted wireless communications with spatial modulation and antenna
  selection,'' \emph{IEEE J. Sel. Areas Commun.}, vol.~38, no.~11, pp.
  2562--2574, 2020.

\bibitem{spim_FL_Elbir2021Jun}
A.~M. Elbir, S.~Coleri, and K.~V. Mishra, ``Federated dropout learning for
  hybrid beamforming with spatial path index modulation in multi-user
  {mmWave-MIMO} systems,'' in \emph{{IEEE International Conference on
  Acoustics, Speech and Signal Processing}}, 2021, pp. 8213--8217.

\bibitem{spim_SM_Clutter_Zhang2021May}
J.~Zhang, ``Clutter mitigation for joint {RadCom} systems based on spatial
  modulation,'' \emph{arXiv preprint arXiv:2105.07328}, 2021.

\bibitem{spim_SM_Clutter2_Zhang2021May}
J.~{Zhang}, ``Beampattern and robust {D}oppler filter design for spatial
  modulation based joint {RadCom} systems,'' \emph{arXiv preprint
  Arxiv:2105.08060}, 2021.

\bibitem{spim_SAIM_LowComple_Zhu2021Jun}
S.~Zhu, F.~Xi, S.~Chen, and A.~Nehorai, ``A low-complexity {MIMO} dual function
  radar communication system via one-bit sampling,'' in \emph{IEEE
  International Conference on Acoustics, Speech and Signal Processing}, 2021,
  pp. 8223--8227.

\bibitem{jrc_generalized_SM_Xu2020Sep}
Z.~Xu, A.~Petropulu, and S.~Sun, ``A joint design of {MIMO-OFDM} dual-function
  radar communication system using generalized spatial modulation,'' in
  \emph{IEEE Radar Conference}, 2020, pp. 1--6.

\bibitem{elbir_BSA_OMP_THZ_CE_Elbir2023Feb}
A.~M. Elbir and S.~Chatzinotas, ``{BSA-OMP: Beam-split-aware orthogonal
  matching pursuit for THz channel estimation},'' \emph{IEEE Wireless Commun.
  Lett.}, p.~1, Feb. 2023.

\bibitem{Gao2023BSISAC}
F.~Gao, L.~Xu, and S.~Ma, ``Integrated sensing and communications with joint
  beam-squint and beam-split for mmwave/thz massive mimo,'' \emph{IEEE
  Transactions on Communications}, vol.~71, no.~5, pp. 2963--2976, May 2023.

\bibitem{elbir_ISAC_antennaSelection_THZ_Elbir2023Jul}
A.~M. Elbir, A.~Abdallah, A.~Celik, and A.~M. Eltawil, ``{Antenna Selection
  With Beam Squint Compensation for Integrated Sensing and Communications},''
  \emph{arXiv}, Jul. 2023.

\bibitem{ummimoTareqOverview}
H.~Sarieddeen, M.-S. Alouini, and T.~Y. Al-Naffouri, ``An overview of signal
  processing techniques for terahertz communications,'' \emph{Proc. IEEE}, vol.
  109, no.~10, pp. 1628--1665, 2021.

\bibitem{ummimoHBThzSVModel}
H.~Yuan, N.~Yang, K.~Yang, C.~Han, and J.~An, ``Hybrid beamforming for
  terahertz multi-carrier systems over frequency selective fading,'' \emph{IEEE
  Trans. Commun.}, vol.~68, no.~10, pp. 6186--6199, 2020.

\bibitem{thz_mmWave_path_Comparison_Yan2020Jun}
L.~Yan, C.~Han, and J.~Yuan, ``{A Dynamic Array-of-Subarrays Architecture and
  Hybrid Precoding Algorithms for Terahertz Wireless Communications},''
  \emph{IEEE J. Sel. Areas Commun.}, vol.~38, no.~9, pp. 2041--2056, 2020.

\bibitem{beampatternDesign_Liu2018Feb}
F.~Liu, C.~Masouros, A.~Li, H.~Sun, and L.~Hanzo, ``{MU-MIMO Communications
  With MIMO Radar: From Co-Existence to Joint Transmission},'' \emph{IEEE
  Trans. Wireless Commun.}, vol.~17, no.~4, pp. 2755--2770, Feb. 2018.

\bibitem{mixedInteger_programming_Burer2012Jul}
S.~Burer and A.~N. Letchford, ``{Non-convex mixed-integer nonlinear
  programming: A survey},'' \emph{Surveys in Operations Research and Management
  Science}, vol.~17, no.~2, pp. 97--106, Jul. 2012.

\bibitem{mimoRadar_WidebandYu2019May}
X.~Yu, G.~Cui, J.~Yang, L.~Kong, and J.~Li, ``Wideband {MIMO} radar waveform
  design,'' \emph{IEEE Trans. Signal Process.}, vol.~67, no.~13, pp.
  3487--3501, 2019.

\bibitem{lfm_ofdm_waveform_Wang2022Dec}
J.~Wang, P.~Wang, F.~Luo, and W.~Wu, ``{Waveform Design and DoA-DoD Estimation
  of OFDM-LFM Signal Based on SDFnT for MIMO Radar},'' \emph{IEEE Access},
  vol.~11, pp. 1348--1358, Dec. 2022.

\bibitem{music}
R.~Schmidt, ``Multiple emitter location and signal parameter estimation,''
  \emph{{IEEE} Trans. Antennas Propag.}, vol.~34, no.~3, pp. 276--280, 1986.

\bibitem{wideband_doaEst_Wideband_Friedlander1993Apr}
B.~Friedlander and A.~J. Weiss, ``{Direction finding for wide-band signals
  using an interpolated array},'' \emph{IEEE Trans. Signal Process.}, vol.~41,
  no.~4, pp. 1618--1634, 1993.

\bibitem{widebandDoAEst_Hybrid_1_Shu2018Feb}
F.~Shu, Y.~Qin, T.~Liu, L.~Gui, Y.~Zhang, J.~Li, and Z.~Han, ``Low-complexity
  and high-resolution {DOA} estimation for hybrid analog and digital massive
  {MIMO} receive array,'' \emph{IEEE Trans. Commun.}, vol.~66, no.~6, pp.
  2487--2501, 2018.

\bibitem{interBeamInterference1_Afeef2022Sep}
L.~Afeef and H.~Arslan, ``{Beam Squint Effect in Multi-Beam mmWave Massive MIMO
  Systems},'' in \emph{{2022 IEEE 96th Vehicular Technology Conference
  (VTC2022-Fall)}}.\hskip 1em plus 0.5em minus 0.4em\relax IEEE, Sep. 2022, pp.
  1--5.

\bibitem{tradeoff_parameterSelection_Chiriyath2015Sep}
A.~R. Chiriyath, B.~Paul, G.~M. Jacyna, and D.~W. Bliss, ``Inner bounds on
  performance of radar and communications co-existence,'' \emph{IEEE Trans.
  Signal Process.}, vol.~64, no.~2, pp. 464--474, 2015.

\bibitem{tradeoff_CPI_Dokhanchi2019Feb}
S.~H. Dokhanchi, B.~S. Mysore, K.~V. Mishra, and B.~Ottersten, ``A {mmWave}
  automotive joint radar-communications system,'' \emph{IEEE Trans. Aerosp.
  Electron. Syst.}, vol.~55, no.~3, pp. 1241--1260, 2019.

\bibitem{fanLiu_waveformDesign_Liu2018Jun}
F.~Liu, L.~Zhou, C.~Masouros, A.~Li, W.~Luo, and A.~Petropulu, ``{Toward
  Dual-functional Radar-Communication Systems: Optimal Waveform Design},''
  \emph{IEEE Trans. Signal Process.}, vol.~66, no.~16, pp. 4264--4279, Jun.
  2018.

\bibitem{hybridBFAltMin}
X.~{Yu}, J.~{Shen}, J.~{Zhang}, and K.~B. {Letaief}, ``{Alternating
  Minimization Algorithms for Hybrid Precoding in Millimeter Wave MIMO
  Systems},'' \emph{{IEEE} J. Sel. Topics Signal Process.}, vol.~10, no.~3, pp.
  485--500, 2016.

\bibitem{mimoRFChainHybrid}
X.~Zhang, A.~F. Molisch, and S.-Y. Kung, ``{Variable-phase-shift-based
  RF-baseband codesign for MIMO antenna selection},'' \emph{{IEEE} Trans.
  Signal Process.}, vol.~53, no.~11, pp. 4091--4103, 2005.

\bibitem{procrustesProblem_Hurley1962Apr}
J.~R. Hurley and R.~B. Cattell, ``{The procrustes program: Producing direct
  rotation to test a hypothesized factor structure},'' \emph{Behav. Sci.},
  vol.~7, no.~2, pp. 258--262, 1962.

\bibitem{amplitudeControl_Lee2020May}
J.~Lee, T.~Oh, J.~Moon, C.~Song, B.~Lee, and I.~Lee, ``Hybrid beamforming with
  variable {RF} attenuator for multi-user {mmWave} systems,'' \emph{IEEE Trans.
  Veh. Technol.}, vol.~69, no.~8, pp. 9131--9134, 2020.

\end{thebibliography}

		
	\end{document}